\documentclass[copyright,creativecommons]{eptcs}
\providecommand{\event}{QPL 2016} 
\usepackage{etex}
\usepackage{pstricks}
\usepackage{amssymb}
\usepackage{amsmath}
\usepackage{amsthm}
\usepackage{stmaryrd}
\usepackage{txfonts}
\usepackage{bbm}
\usepackage{enumerate}
\usepackage{multicol}
\usepackage{rotate}

\theoremstyle{plain}
\newtheorem{theorem}{Theorem}[section]
\newtheorem{proposition}[theorem]{Proposition}
\newtheorem{lemma}[theorem]{Lemma}
\newtheorem{corollary}[theorem]{Corollary}

\theoremstyle{definition}
\newtheorem{definition}{Definition}[section]

\theoremstyle{remark}

\newdimen\proofrulebreadth \proofrulebreadth=.05em
\newdimen\proofdotseparation \proofdotseparation=1.25ex
\newdimen\proofrulebaseline \proofrulebaseline=2ex
\newcount\proofdotnumber \proofdotnumber=3
\let\then\relax
\def\hfi{\hskip0pt plus.0001fil}
\mathchardef\squigto="3A3B
\newif\ifinsideprooftree\insideprooftreefalse
\newif\ifonleftofproofrule\onleftofproofrulefalse
\newif\ifproofdots\proofdotsfalse
\newif\ifdoubleproof\doubleprooffalse
\let\wereinproofbit\relax
\newdimen\shortenproofleft
\newdimen\shortenproofright
\newdimen\proofbelowshift
\newbox\proofabove
\newbox\proofbelow
\newbox\proofrulename
\def\shiftproofbelow{\let\next\relax\afterassignment\setshiftproofbelow\dimen0 }
\def\shiftproofbelowneg{\def\next{\multiply\dimen0 by-1 }%
\afterassignment\setshiftproofbelow\dimen0 }
\def\setshiftproofbelow{\next\proofbelowshift=\dimen0 }
\def\setproofrulebreadth{\proofrulebreadth}

\def\prooftree{
\ifnum  \lastpenalty=1
\then   \unpenalty
\else   \onleftofproofrulefalse
\fi
\ifonleftofproofrule
\else   \ifinsideprooftree
        \then   \hskip.5em plus1fil
        \fi
\fi
\bgroup
\setbox\proofbelow=\hbox{}\setbox\proofrulename=\hbox{}%
\let\justifies\proofover\let\leadsto\proofoverdots\let\Justifies\proofoverdbl
\let\using\proofusing\let\[\prooftree
\ifinsideprooftree\let\]\endprooftree\fi
\proofdotsfalse\doubleprooffalse
\let\thickness\setproofrulebreadth
\let\shiftright\shiftproofbelow \let\shift\shiftproofbelow
\let\shiftleft\shiftproofbelowneg
\let\ifwasinsideprooftree\ifinsideprooftree
\insideprooftreetrue
\setbox\proofabove=\hbox\bgroup$\displaystyle 
\let\wereinproofbit\prooftree
\shortenproofleft=0pt \shortenproofright=0pt \proofbelowshift=0pt
\onleftofproofruletrue\penalty1
}

\def\eproofbit{
\ifx    \wereinproofbit\prooftree
\then   \ifcase \lastpenalty
        \then   \shortenproofright=0pt  
        \or     \unpenalty\hfil         
        \or     \unpenalty\unskip       
        \else   \shortenproofright=0pt  
        \fi
\fi
\global\dimen0=\shortenproofleft
\global\dimen1=\shortenproofright
\global\dimen2=\proofrulebreadth
\global\dimen3=\proofbelowshift
\global\dimen4=\proofdotseparation
\global\count255=\proofdotnumber
$\egroup  
\shortenproofleft=\dimen0
\shortenproofright=\dimen1
\proofrulebreadth=\dimen2
\proofbelowshift=\dimen3
\proofdotseparation=\dimen4
\proofdotnumber=\count255
}

\def\proofover{
\eproofbit 
\setbox\proofbelow=\hbox\bgroup 
\let\wereinproofbit\proofover
$\displaystyle
}%
\def\proofoverdbl{
\eproofbit 
\doubleprooftrue
\setbox\proofbelow=\hbox\bgroup 
\let\wereinproofbit\proofoverdbl
$\displaystyle
}%
\def\proofoverdots{
\eproofbit 
\proofdotstrue
\setbox\proofbelow=\hbox\bgroup 
\let\wereinproofbit\proofoverdots
$\displaystyle
}%
\def\proofusing{
\eproofbit 
\setbox\proofrulename=\hbox\bgroup 
\let\wereinproofbit\proofusing
\kern0.3em$
}

\def\endprooftree{
\eproofbit 
  \dimen5 =0pt
\dimen0=\wd\proofabove \advance\dimen0-\shortenproofleft
\advance\dimen0-\shortenproofright
\dimen1=.5\dimen0 \advance\dimen1-.5\wd\proofbelow
\dimen4=\dimen1
\advance\dimen1\proofbelowshift \advance\dimen4-\proofbelowshift
\ifdim  \dimen1<0pt
\then   \advance\shortenproofleft\dimen1
        \advance\dimen0-\dimen1
        \dimen1=0pt
        \ifdim  \shortenproofleft<0pt
        \then   \setbox\proofabove=\hbox{%
                        \kern-\shortenproofleft\unhbox\proofabove}%
                \shortenproofleft=0pt
        \fi
\fi
\ifdim  \dimen4<0pt
\then   \advance\shortenproofright\dimen4
        \advance\dimen0-\dimen4
        \dimen4=0pt
\fi
\ifdim  \shortenproofright<\wd\proofrulename
\then   \shortenproofright=\wd\proofrulename
\fi
\dimen2=\shortenproofleft \advance\dimen2 by\dimen1
\dimen3=\shortenproofright\advance\dimen3 by\dimen4
\ifproofdots
\then
        \dimen6=\shortenproofleft \advance\dimen6 .5\dimen0
        \setbox1=\vbox to\proofdotseparation{\vss\hbox{$\cdot$}\vss}%
        \setbox0=\hbox{%
                \advance\dimen6-.5\wd1
                \kern\dimen6
                $\vcenter to\proofdotnumber\proofdotseparation
                        {\leaders\box1\vfill}$%
                \unhbox\proofrulename}%
\else   \dimen6=\fontdimen22\the\textfont2 
        \dimen7=\dimen6
        \advance\dimen6by.5\proofrulebreadth
        \advance\dimen7by-.5\proofrulebreadth
        \setbox0=\hbox{%
                \kern\shortenproofleft
                \ifdoubleproof
                \then   \hbox to\dimen0{%
                        $\mathsurround0pt\mathord=\mkern-6mu%
                        \cleaders\hbox{$\mkern-2mu=\mkern-2mu$}\hfill
                        \mkern-6mu\mathord=$}%
                \else   \vrule height\dimen6 depth-\dimen7 width\dimen0
                \fi
                \unhbox\proofrulename}%
        \ht0=\dimen6 \dp0=-\dimen7
\fi
\let\doll\relax
\ifwasinsideprooftree
\then   \let\VBOX\vbox
\else   \ifmmode\else$\let\doll=$\fi
        \let\VBOX\vcenter
\fi
\VBOX   {\baselineskip\proofrulebaseline \lineskip.2ex
        \expandafter\lineskiplimit\ifproofdots0ex\else-0.6ex\fi
        \hbox   spread\dimen5   {\hfi\unhbox\proofabove\hfi}%
        \hbox{\box0}%
        \hbox   {\kern\dimen2 \box\proofbelow}}\doll%
\global\dimen2=\dimen2
\global\dimen3=\dimen3
\egroup 
\ifonleftofproofrule
\then   \shortenproofleft=\dimen2
\fi
\shortenproofright=\dimen3
\onleftofproofrulefalse
\ifinsideprooftree
\then   \hskip.5em plus 1fil \penalty2
\fi
}

\renewcommand{\to}{\xrightarrow{}}
\newcommand{\tto}[1]{\xrightarrow{#1}}
\newcommand{\oot}[1]{\xleftarrow{#1}}

\newcommand{\ttwoot}[2]{\underset{#2}{\overset{#1}{\leftleftarrows}}}

\newcommand{\Pfn}{\mathsf{Pfn}}

\newcommand{\Rel}{\mathsf{Rel}}

\newcommand{\FHilb}{\mathsf{FHilb}}

\newcommand{\id}{\mathrm{id}}

\newcommand{\CCC}{{\cal C}}

\renewcommand{\Bbb}{\mathbb}

\newcommand{\CCc}{{\Bbb C}}

\newcommand{\ZZz}{{\Bbb Z}}

\mathcode`\<="4268 
\mathcode`\>="5269 
\mathchardef\gt="313E 
\mathchardef\lt="313C 

\newcommand{\beq}{\begin{equation}}
\newcommand{\eeq}{\end{equation}}
\newcommand{\ba}[1]{\begin{array}{#1}}
\newcommand{\ea}{\end{array}}
\newcommand{\bea}{\begin{eqnarray}}
\newcommand{\eea}{\end{eqnarray}}
\newcommand{\bear}{\begin{eqnarray*}}
\newcommand{\eear}{\end{eqnarray*}}

\newcommand{\conv}[2]{_{#1}{\star}_{#2}}
\newcommand{\convv}{{\star}}
\newcommand{\mnd}{{\scriptstyle \blacktriangledown}}
\newcommand{\unt}{\mbox{!`}}
\newcommand{\cmn}{{\scriptstyle \blacktriangle}}
\newcommand{\cun}{{
 !}}
\newcommand{\ortt}{\neg}
\newcommand{\ort}[1]{\neg{#1}}

\title{(Modular) Effect Algebras are Equivalent to\\ (Frobenius)  Antispecial Algebras
}
\author{Dusko Pavlovic
\institute{University of Hawaii, Honolulu HI}
\email{dusko@hawaii.edu}
\and
Peter-Michael Seidel
\institute{University of Hawaii, Honolulu HI}
\email{pseidel@hawaii.edu}
}
\def\titlerunning{(Modular) Effect =  (Frobenius) Antispecial}
\def\authorrunning{D. Pavlovic and P.-M. Seidel}

\begin{document}
\maketitle

\begin{abstract}
Effect algebras are one of the generalizations of Boolean algebras proposed in the quest for a \emph{quantum logic}. Frobenius algebras are a tool of \emph{categorical quantum mechanics}, used to present various families of observables in abstract, often nonstandard frameworks. Both effect algebras and Frobenius algebras capture their respective fragments of quantum mechanics by elegant and succinct axioms; and both come with their conceptual mysteries. A particularly elegant and mysterious constraint, imposed on Frobenius algebras to characterize a class of tripartite entangled states, is the \emph{antispecial}\/ law. A particularly contentious issue on the quantum logic side is the \emph{modularity}\/ law, proposed by von Neumann to mitigate the failure of distributivity of quantum logical connectives. We show that, if quantum logic and categorical quantum mechanics are formalized in the same framework, then the antispecial law of categorical quantum mechanics corresponds to the natural requirement of effect algebras that the units are each other's unique complements; and that the modularity law corresponds to the Frobenius condition. These correspondences lead to the equivalence announced in the title. Aligning the two formalisms, at the very least, sheds new light on the concepts that are more clearly displayed on one side than on the other (such as e.g. the orthogonality). Beyond that, it may also open up new approaches to deep and important problems of quantum mechanics (such as the classification of complementary observables). 
\end{abstract}

\section{Introduction}

That \emph{"nobody understands quantum mechanics"}\/ (as Richard Feynman announced) may be the state of the world. That the standard mathematical formalisms of quantum mechanics contain features that do not correspond to any features of their subject (as John von Neumann pointed out \cite{RedeiM:why} almost immediately after he published his treatise \cite{Neumann:foundations} about those mathematical formalisms) is definitely a social phenomenon. Von Neumann attacked the problem, and generated \emph{quantum logics}  \cite{Neumann:continuous,Birkhoff-vonNeumann:LQM}, which became a popular research area of lattice theory.  Many years later, mathematicians and computer scientists attacked the same problem, and generated \emph{categorical quantum mechanics} \cite{Abramsky-Coecke,SelingerP:CPM,Coecke-Duncan:ZX,CoeckeB:ComQuantLog}, which became a popular research area of category theory. Most recently, an ambitious effort has been initiated to incorporate both families of structures, and much more, under a new structure called \emph{effectus}  \cite{JacobsB:NewDir,JacobsB:effectus}. The present note is, of course, incomparable with that effort in its scope, but it also attempts to relate two families of structures, one from quantum logic, the other one from categorical quantum mechanics, and is thus concerned with a closely related conceptual bridge. Being much smaller, our bridge does not require any new material: we simply translate between the two languages, and try to align the concepts underlying the different models that turn out to be structurally equivalent.

More precisely, we relate the realm of effect algebras \cite{Beltrametti-Bugajski,Foulis-Bennett,GudderS:effect}, intended to capture quantum propositions just like Boolean algebras capture classical propositions, and the realm of Frobenius algebras \cite{Carboni-Walters,PavlovicD:QMWS,PavlovicD:CQStruct,PavlovicD:Qabs12,PavlovicD:MSCS13}, used to capture classical data in a quantum universe, viewed as a category. Although the two research programs have been driven by different goals and realized by substantially different mathematical methods, they turn out to lead to equivalent structural components. Understanding this equivalence means uncovering the common conceptual components underlying both theories. Instantiating Frobenius algebras to the category $\Rel$ of sets and relations, and generalizing effect algebras to an abstract dagger compact category $\CCc$, we get the equivalences announced in the title of the paper.

\subsubsection*{Outline of the paper} We begin by defining effect algebras in Sec.~\ref{Sec:effect}. As usual, effect algebras are defined as sets with some partial operations, but the defining conditions are formalized in categorical terms, since our goal is to align them with the similar conditions that arise in categorical quantum mechanics. Towards this goal, in the rest of the paper we work with an abstract dagger compact category $\CCc$. The original definition of effect algebras is recovered for $\CCc = \Rel$, the category of sets and relations. Since its compactness and the self-dualities of its objects are an important tool of the analysis, the restriction to partial maps, prominent in the definition of effect algebras, is not hardwired in the definition of the environment category, but imposed in the definition of the analyzed structures. Before we get to that restriction, we analyze the general operation of orthocomplementation in general terms of dagger compactness in Sec.~\ref{Sec:Ortho},. The reasons and the tools for the restriction to partial maps are discussed in Sec.~\ref{Sec:Anti-effect}. The tools boil down to a small fragment of the categorical theory of maps, described in Sec.~\ref{Sec:Maps}, relative to the convolution operations in Sec.~\ref{Sec:Convol}. In Sec.~\ref{Sec:superspec}, we finally reach the stage where we can propose a categorical version of the effect algebra structure. The claim is that the special and the antispecial requirements, that play an interesting role in categorical quantum mechanics, in fact capture the same structure as effect algebras. The main claim is Prop.~\ref{prop:eff}, which says that special and antispecial algebras (christened \emph{superspecial}\/ for this occasion) are just those that satisfy the categorical definition of effect algebras, simply lifted from sets and partial functions to dagger compact categories. The technical gain from this characterization is that the superspecial strucutre is a standard piece of categorical algebra, well oiled for diagrammatic analyses in categorical quantum mechanics, whereas the categorical version of the original definition of effect algebras involves pullbacks, and requires subtle and often cumbersome arguments, as illustrated already in the proof of Prop.~\ref{prop:eff}. Finally, in Sec.~\ref{Sec:Frob}, we show that the modularity law, satisfied by some effect algebras, corresponds  to the Frobenius law in superspecial algebras. This not only connects two laws that are studied extensively in two research areas, but also generalizes the concept of modularity from sets to dagger compact categories, while providing an intuitive view of the Frobenius law. In Sec.~\ref{Sec:Conclusion}, we comment about applications of the results and about further work suggested by the results.

%
%
%
%
%

\section{Effect algebras
}\label{Sec:effect}

\paragraph{Background.} Effect algebras \cite{Beltrametti-Bugajski,Foulis-Bennett,GudderS:effect} are an offshoot of the effort towards generalizing classical propositional logic into a putative quantum logic, initiated by von Neumann \cite{Neumann:continuous,Birkhoff-vonNeumann:LQM}. The effort never led to a logical system in the traditional sense, perhaps because the deduction and abstraction mechanisms that the logicians use to define such systems, actually characterize classical data in a quantum universe, whereas quantum data disobey such abstraction mechanisms by their very nature \cite{PavlovicD:Qabs12}. At the propositional level, these abstraction mechanisms manifest themselves as the distributivity laws. Without such laws, quantum logics remained as unintuitive for the logicians as quantum physics has been for the physicists. This provided a business opportunity for some mathematicians and philosophers. Effect algebras are a result of this opportunity.

\paragraph{Idea.} Quantum propositions, viewed as the elements of an effect algebra, can be thought of as subspaces of a Hilbert space. They are operated on by the quantum logical connectives $\ovee, \owedge$ and $\ortt$, which are analogous to the classical disjunction $\vee$, conjunction $\wedge$ and negation $\neg$. The difference is that any two classical propositions $p$ and $q$ can be composed into $p\vee q$, $p \wedge q$, whereas the quantum propositions $u$ and $v$ can only be composed into $u \ovee v$, $u\owedge v$ if the corresponding Hilbert subspaces are orthogonal; otherwise these compositions are undefined. The complements $\ort u$ are always defined. The partiality of the quantum logical connectives $\ovee$ and $\owedge$ is induced by the fact that non-orthogonal quantum states cannot be reliably distinguished,  which  implies that quantum observables, which are denoted by quantum propositions, and reasoned about in quantum logic, can only be formed from orthogonal Hilbert subspaces. Effect algebras thus attempt to capture the essence of quantum logic in terms of \emph{partiality}\/ of quantum logical operations.

\begin{definition}\label{def:effect}
An \emph{effect algebra} is a set $A$ together with the  partial functions 
\beq\label{eq:one}
A\times A\tto\ovee A \oot\ortt A \ttwoot 0 1
I \eeq
where $I$ is a singleton set, and moreover 
\begin{itemize}
\item $(A,\ovee,0)$ is a commutative monoid, 
\item the following conditions are satisified for all $x,y\in A$
\bea 
x\ovee y = 1  & \iff & x = \ortt y \label{eq:two}\\ 
x\ovee 1 = 1  & \iff & x = 0 \label{eq:three}
\eea
\end{itemize}
\end{definition}


\paragraph{Remarks.} It is easy to see that the above definition is equivalent with the original one in \cite{Foulis-Bennett}. Proving that $\ort{\ort x} = x$, that the partial elements $0,1:I\to A$ must be total, and that $\ortt$ must be a map (total and single-valued\footnote{Here we use \emph{maps}, or \emph{functions}, defined as total and single-valued relations in basic set theory. In Sec.~\ref{Sec:Maps} we shall see how these definitions extend to much more general categorical frameworks, including dagger-compact categories with classical structures.} are instructive exercises.

A category theorist might interpret the above definition by viewing the effect algebra signature, displayed in \eqref{eq:one}, as a diagram in the category $\Pfn$ of sets and partial maps. The requirement that $(A,\ovee,0)$ is a commutative monoid is expressed by familiar commutative diagrams, and conditions \eqref{eq:two} and \eqref{eq:three} mean that the following squares must be pullbacks in $\Pfn$.
\beq\label{eq:pb}
\newcommand{\Unit}{\scriptstyle I}
\newcommand{\Carrier}{\scriptstyle A}
\newcommand{\Pairs}{\scriptstyle A\otimes A}
\newcommand{\oone}{\scriptstyle 1}
\newcommand{\ttwo}{\scriptstyle <\id, \ortt>}
\newcommand{\uup}{\scriptstyle\cun}
\newcommand{\ddown}{\scriptstyle\ovee}
\newcommand{\zzero}{\scriptstyle 0}
\newcommand{\zzeros}{\scriptstyle<0,0>}
\newcommand{\hiid}{\scriptstyle \id}
\def\JPicScale{.6}
\ifx\JPicScale\undefined\def\JPicScale{1}\fi
\psset{unit=\JPicScale mm}
\psset{linewidth=0.3,dotsep=1,hatchwidth=0.3,hatchsep=1.5,shadowsize=1,dimen=middle}
\psset{dotsize=0.7 2.5,dotscale=1 1,fillcolor=black}
\psset{arrowsize=1 2,arrowlength=1,arrowinset=0.25,tbarsize=0.7 5,bracketlength=0.15,rbracketlength=0.15}
\begin{pspicture}(0,0)(111.25,16.25)
\pscustom[]{\psline{<-}(32.5,15)(7.5,15)
\psline(7.5,15)(10,15)
\psbezier{-}(10,15)(10,15)(10,15)
}
\pscustom[]{\psline{<-}(32.5,-15)(13.12,-15)
\psbezier(13.12,-15)(13.12,-15)(13.12,-15)
\psbezier{-}(13.12,-15)(13.12,-15)(13.12,-15)
}
\pscustom[]{\psline{<-}(35,-10)(35,10)
\psbezier(35,10)(35,10)(35,10)
\psbezier{-}(35,10)(35,10)(35,10)
}
\pscustom[]{\psline{<-}(5,-10)(5,10)
\psbezier(5,10)(5,10)(5,10)
\psbezier{-}(5,10)(5,10)(5,10)
}
\rput(35,15){$\Unit$}
\rput(35,-15){$\Carrier$}
\rput(5,-15){$\Pairs$}
\rput(5,15){$\Carrier$}
\rput[l](35.62,0){$\oone$}
\rput[r](4.38,0){$\ttwo$}
\rput[b](20,16.25){$\uup$}
\rput[t](20,-16.25){$\ddown$}
\pscustom[]{\psline{<-}(107.5,15)(82.5,15)
\psline(82.5,15)(85,15)
\psbezier{-}(85,15)(85,15)(85,15)
}
\psline{<-}(107.5,-15)
(88.12,-15)
(88.75,-15)(87.5,-15)
\pscustom[]{\psline{<-}(110,-10)(110,10)
\psbezier(110,10)(110,10)(110,10)
\psbezier{-}(110,10)(110,10)(110,10)
}
\pscustom[]{\psline{<-}(80,-10)(80,10)
\psbezier(80,10)(80,10)(80,10)
\psbezier{-}(80,10)(80,10)(80,10)
}
\rput(110,15){$\Unit$}
\rput(110,-15){$\Carrier$}
\rput(80,-15){$\Pairs$}
\rput[t](95,-16.25){$\ddown$}
\pscustom[]{\psline(10,10)(10,12.5)
\psbezier(10,12.5)(10,12.5)(10,12.5)
\psbezier(10,12.5)(10,12.5)(10,12.5)
}
\pscustom[]{\psline(10,10)(7.5,10)
\psbezier(7.5,10)(7.5,10)(7.5,10)
\psbezier(7.5,10)(7.5,10)(7.5,10)
}
\pscustom[]{\psline(85,10)(85,12.5)
\psbezier(85,12.5)(85,12.5)(85,12.5)
\psbezier(85,12.5)(85,12.5)(85,12.5)
}
\pscustom[]{\psline(85,10)(82.5,10)
\psbezier(82.5,10)(82.5,10)(82.5,10)
\psbezier(82.5,10)(82.5,10)(82.5,10)
}
\rput[l](111.25,0){$\zzero$}
\rput[r](79.38,0){$\zzeros$}
\rput(80,15){$\Unit$}
\rput[b](94.38,15.62){$\hiid$}
\end{pspicture}

\eeq
\medskip

\noindent The tensors and the pairing are induced by the cartesian products of sets. The arrow $\cun :A\to I$ is the map sending all elements of $A$ into the singleton element of $I$. While the left-hand pullback is easily seen to capture  \eqref{eq:two}, the right hand pullback actually says that $x\ovee y = 0 \iff x = 0 = y$, or that the monoid is torsion-free, which is equivalent with \eqref{eq:three} because \eqref{eq:two} implies that $x\ovee 1 = y \iff  x\ovee 1 \ovee \ort y = 1\iff x \ovee \ort y = 0$.

But a categorical quantum mechanic might be inclined to go even further, and draw the the above pullbacks as string diagrams:
\beq\label{eq:pb-string}
\newcommand{\oovee}{\ovee}
\newcommand{\Carrier}{\scriptstyle A}
\newcommand{\Juu}{\scriptstyle U}
\newcommand{\Jvee}{\scriptstyle V}
\newcommand{\uu}{\scriptstyle \exists ! u}
\newcommand{\oone}{\scriptscriptstyle1}
\newcommand{\uuone}{\scriptstyle \forall u_0}
\newcommand{\uutwo}{\scriptstyle \forall u_1}
\newcommand{\EQLS}{=}
\newcommand{\zzero}{\scriptscriptstyle 0}
\newcommand{\uuthree}{\scriptstyle \forall v_1}
\newcommand{\anyvee}{\scriptstyle \forall v_0}
\def\JPicScale{.6}
\ifx\JPicScale\undefined\def\JPicScale{1}\fi
\psset{unit=\JPicScale mm}
\psset{linewidth=0.3,dotsep=1,hatchwidth=0.3,hatchsep=1.5,shadowsize=1,dimen=middle}
\psset{dotsize=0.7 2.5,dotscale=1 1,fillcolor=black}
\psset{arrowsize=1 2,arrowlength=1,arrowinset=0.25,tbarsize=0.7 5,bracketlength=0.15,rbracketlength=0.15}
\begin{pspicture}(0,0)(218.13,38.12)
\psline[linewidth=0.35,border=1.05](6.25,22.5)
(6.25,0.62)(13.12,-3.75)
\psline[linewidth=0.35](20,22.5)
(20,0.62)(13.75,-3.12)
\rput{90}(13.12,-3.75){\psellipse[linewidth=0.35,fillstyle=solid](0,0)(1.49,-1.48)}
\psline[linewidth=0.35](13.12,-30.62)(13.12,-4.38)
\pscustom[linewidth=0.2]{\psline(0,14.38)(0,-23.75)
\psline(0,-23.75)(26.88,-23.75)
\psline(26.88,-23.75)(26.88,14.38)
\psline(26.88,14.38)(0,14.38)
\psline(0,14.38)(0,11.25)
\psbezier(0,11.25)(0,11.25)(0,11.25)
\psline(0,11.25)(0,12.5)
}
\pscustom[linewidth=0.35]{\psline(25.63,22.49)(1.87,22.5)
\psline(1.87,22.5)(10,30.62)
\psbezier(10,30.62)(10,30.62)(10,30.62)
\psline(10,30.62)(17.5,30.62)
\psline(17.5,30.62)(25.63,22.49)
\psbezier(25.63,22.49)(25.63,22.49)(25.63,22.49)
\psline(25.63,22.49)(25.62,22.49)
}
\psline[linewidth=0.2](46.88,-20.62)
(68.12,0.62)
(88.75,-20)
(47.5,-20)(48.12,-19.38)
\psline[linewidth=0.35](13.75,30.62)(13.75,38.12)
\pscustom[linewidth=0.35,fillcolor=white,fillstyle=solid]{\psbezier(17.5,5.62)(17.5,5.62)(17.5,5.62)(17.5,5.62)
\psline(17.5,5.62)(23.12,5.62)
\psline(23.12,5.62)(23.12,11.24)
\psline(23.12,11.24)(17.5,11.24)
\psline(17.5,11.24)(17.5,5.62)
\psbezier(17.5,5.62)(17.5,5.62)(17.5,5.62)
\psline(17.5,5.62)(17.5,6.88)
}
\rput(20,8.12){$\ortt$}
\pscustom[linewidth=0.35,linecolor=red,fillcolor=white,fillstyle=solid]{\psline(8.12,-21.87)(8.12,-21.88)
\psline(8.12,-21.88)(18.12,-21.87)
\psline(18.12,-21.87)(18.12,-16.25)
\psline(18.12,-16.25)(8.12,-16.25)
\psline(8.12,-16.25)(8.12,-21.88)
\psbezier(8.12,-21.88)(8.12,-21.88)(8.12,-21.88)
\psline(8.12,-21.88)(9.38,-21.87)
}
\psline[linewidth=0.35](68.12,-30.62)(68.12,-6.88)
\pscustom[linewidth=0.35,linecolor=red,fillcolor=white,fillstyle=solid]{\psline(63.12,-18.74)(63.12,-18.76)
\psline(63.12,-18.76)(73.12,-18.74)
\psline(73.12,-18.74)(73.12,-13.12)
\psline(73.12,-13.12)(63.12,-13.12)
\psline(63.12,-13.12)(63.12,-18.76)
\psbezier(63.12,-18.76)(63.12,-18.76)(63.12,-18.76)
\psline(63.12,-18.76)(64.38,-18.74)
}
\rput{0}(68.12,-5.62){\psellipse[linewidth=0.35,fillstyle=solid](0,0)(1.49,-1.48)}
\psline[linewidth=0.35](68.12,10.62)(68.12,36.88)
\rput(44.38,1.25){$\EQLS$}
\rput(68.12,-16.26){$\uu$}
\rput(13.12,-19.38){$\uu$}
\rput(23.75,-21.25){$\uuone$}
\rput(81.25,-17.5){$\uutwo$}
\rput(13.75,26.25){$\oovee$}
\rput(68.12,8.12){$\oone$}
\rput[l](14.38,-11.88){$\Carrier$}
\rput[r](66.88,-10.62){$\Carrier$}
\rput[l](15,35){$\Carrier$}
\rput[l](69.38,35){$\Carrier$}
\rput[l](21.25,18.12){$\Carrier$}
\rput[r](5,18.12){$\Carrier$}
\psline[linewidth=0.35](146.25,22.5)
(146.25,10.62)(146.25,11.25)
\psline[linewidth=0.35](153.12,-30)(153.12,-10.62)
\pscustom[linewidth=0.2]{\psline(140,14.38)(140,-20)
\psline(140,-20)(166.88,-20)
\psline(166.88,-20)(166.88,14.38)
\psline(166.88,14.38)(140,14.38)
\psline(140,14.38)(140,11.25)
\psbezier(140,11.25)(140,11.25)(140,11.25)
\psline(140,11.25)(140,12.5)
}
\pscustom[linewidth=0.35]{\psline(165.63,22.49)(141.87,22.5)
\psline(141.87,22.5)(150,30.62)
\psbezier(150,30.62)(150,30.62)(150,30.62)
\psline(150,30.62)(157.5,30.62)
\psline(157.5,30.62)(165.63,22.49)
\psbezier(165.63,22.49)(165.63,22.49)(165.63,22.49)
\psline(165.63,22.49)(165.62,22.49)
}
\psline[linewidth=0.35](153.75,30.62)(153.75,38.12)
\psline[linewidth=0.35](208.12,-30)(208.12,-10)
\psline[linewidth=0.35](208.12,10)(208.12,38.12)
\rput(183.75,1.88){$\EQLS$}
\rput(153.75,26.25){$\oovee$}
\rput[l](155,35){$\Carrier$}
\rput[l](209.38,35){$\Carrier$}
\rput[l](161.25,18.12){$\Carrier$}
\rput[r](145,18.12){$\Carrier$}
\psline[linewidth=0.35,linecolor=red](143.12,-10.62)
(153.12,-0.62)
(163.12,-10.62)
(143.12,-10.62)(153.12,-10.62)
\pscustom[linewidth=0.35]{\psline(151.88,10.62)(146.25,5)
\psline(146.25,5)(140.62,10.62)
\psline(140.62,10.62)(151.88,10.62)
\psbezier(151.88,10.62)(151.88,10.62)(151.88,10.62)
}
\pscustom[linewidth=0.35]{\psline(160,21.88)(160,10.62)
\psbezier(160,10.62)(160,10.62)(160,10.62)
}
\psline[linewidth=0.35](198.13,-10)
(208.12,0)
(218.13,-10)
(198.13,-10)(208.13,-10)
\pscustom[linewidth=0.35]{\psline(165.62,10.62)(160,5)
\psline(160,5)(154.38,10.62)
\psline(154.38,10.62)(165.62,10.62)
\psbezier(165.62,10.62)(165.62,10.62)(165.62,10.62)
}
\pscustom[linewidth=0.35]{\psline(73.75,10.62)(68.12,5)
\psline(68.12,5)(62.5,10.62)
\psline(62.5,10.62)(73.75,10.62)
\psbezier(73.75,10.62)(73.75,10.62)(73.75,10.62)
}
\pscustom[linewidth=0.35]{\psline(213.75,10)(208.12,4.38)
\psline(208.12,4.38)(202.5,10)
\psline(202.5,10)(213.75,10)
\psbezier(213.75,10)(213.75,10)(213.75,10)
}
\rput(208.12,-6.25){$\uuthree$}
\rput(153.12,-7.5){$\uuthree$}
\rput(163.75,-16.88){$\anyvee$}
\rput(208.12,7.5){$\zzero$}
\rput(160,8.12){$\zzero$}
\rput(146.25,8.12){$\zzero$}
\rput[l](13.75,-29.38){$\Juu$}
\rput[l](69.38,-28.75){$\Juu$}
\rput[l](154.38,-29.38){$\Jvee$}
\rput[l](209.38,-28.75){$\Jvee$}
\pscustom[linewidth=0.35]{\psline(62.5,-8.12)(68.13,-2.5)
\psline(68.13,-2.5)(73.75,-8.12)
\psline(73.75,-8.12)(62.5,-8.12)
\psbezier(62.5,-8.12)(62.5,-8.12)(62.5,-8.12)
}
\pscustom[linewidth=0.35]{\psline(1.87,0.63)(25.63,0.61)
\psline(25.63,0.61)(17.5,-7.5)
\psbezier(17.5,-7.5)(17.5,-7.5)(17.5,-7.5)
\psline(17.5,-7.5)(10,-7.5)
\psline(10,-7.5)(1.87,0.63)
\psbezier(1.87,0.63)(1.87,0.63)(1.87,0.63)
\psline(1.87,0.63)(1.88,0.63)
}
\end{pspicture}

\eeq
\vspace{3\baselineskip}

\noindent The left-hand diagram should be read as saying that for every $u_0: U\to A\otimes A$ and $u_1: U\to I$ such that $\ovee \circ u_0 = 1 \circ u_1$, there is a unique $u:U\to A$ with $u_0 = <\id,\ortt>\circ u$ and $u_1 = \cun \circ u$.  The left-hand diagram in \eqref{eq:pb-string} just says in string diagrams that the left-hand square in \eqref{eq:pb} is a pullback. The right-hand diagram in \eqref{eq:pb-string} says that the right-hand square in  in \eqref{eq:pb} is a pullback, and it should be read as saying  that every $v_0: V\to A\otimes A$ and $v_1: V\to I$ such that $\ovee \circ v_0 = 0 \circ v_1$, must satisfy $v_0 = <0,0>\circ v_1$. The unique pullback factorization must be $v_1$, because the top side of the right-hand square in \eqref{eq:pb} is the identity.

If these conditions are accepted as a high level view of the "propositional" operations on quantum observables, then it is natural to ask what they mean beyond mere partial functions, in the categories with more features of quantum mechanics. The immediate obstacle to a straightforward lifting of the above definition is that most categories that we encounter in categorical quantum mechanics lack most pullbacks, and conditions \eqref{eq:pb} do not capture the intended meaning. Restricting them to abstract partial functions, like we shall do in Def.~\ref{def:effect-gen}, narrows the meaning of \eqref{eq:pb}, but the pullbacks remain inconvenient to work with. The relief comes from a surprising direction: the pullbacks requirements in \eqref{eq:pb} and \eqref{eq:pb-string} turn out to be equivalent to some more convenient conditions, independently encountered in categorical quantum mechanics.

\section{Orthocomplemented algebras
}\label{Sec:Ortho}

\subsection{Dagger-compact categories and classical structures}\label{Sec:cqm}
While effect algebras are normally presented as sets with essentially algebraic structure\footnote{An algebraic structure is presented by operations and equations. An \emph{essentially}\/ algebraic structure is presented by operations and \emph{conditional}\/ equations, which are the statements in the form $p\Rightarrow q$, and $p$ and $q$ are equations. Besides effect algebras, the examples of essentially algebraic structures include categories and the varieties of categorical algebra, defined by algebraic theories using functors and natural transformations \cite{AdamekJ:locpac}.}, we now broaden the scope, and study the components of their structure in  the abstract framework of a \emph{dagger-compact category} $\CCc$. The standard definition of effect algebras will be recovered as the special case where $\CCc = \Rel$, the category of sets and relations, concrete or abstract \cite{Carboni-Walters,Freyd-Scedrov:book}, as used in \cite{GogiosoS:relations,HeunenC:relations,PavlovicD:QI09,PavlovicD:Qabs12}.

The idea of lifting the effect algebra structure beyond sets, and expressing it in abstract categorical terms, is that studying the effect algebra operations in other models of quantum mechanics, standard and nonstandard \cite{PavlovicD:QPL09}, will reveal their relationships with other quantum operations and axiomatizations. For instance, it seems interesting to ask what is the suitable notion of effect algebra in the framework of Hilbert spaces. Although the effect algebra operations were conceptualized as an abstraction of the relevant "propositional" operations over the families of orthogonal subspaces of a Hilbert space, it is remarkable that these operations are not expressible in the language of Hilbert spaces themselves, or even in terms of categorical operations over  Hilbert spaces. To see this, note that, the category of Hilbert spaces has very few pullbacks, and that lifting the pullbacks \eqref{eq:pb}, or \eqref{eq:pb-string} to Hilbert spaces does not give usable requirements.

Recall that dagger-compact categories are just compact (closed) categories, going back all the way to \cite{Kelly-Laplaza}, but extended with an additional duality, the \emph{dagger}\/ functor $\ddag: \CCc^{o}\to \CCc$, which commutes with the compact duality $\ast :\CCc^{o} \to \CCc$ up to an coherent isomorphism $X^{\ast\ddag} \cong X^{\ddag\ast}$. The standard model is the category of finite-dimensional complex Hilbert spaces $\FHilb$. One of the main points of working with an abstract categorical signature, rather than with concrete Hilbert spaces, is that nonstandard and toy models \cite{AbramskyS:BigToy,Foulis-Randall:test,Mermin:moon,PavlovicD:QPL09,SpekkensR:toy} often provide important information. Another point, going back to von Neumann, is that many features of the Hilbert space structure do not correspond to any features of quantum mechanics that they are used to describe.\footnote{In terms of categorical semantics, this means that the Hilbert space model is not \emph{fully abstract}:  it always displays some "irrelevant implementation details" \cite{MilnerR:fully-abstract}.} Presented in terms of the functor $X_\ast = X^{\ast\ddag}$, and equipped with the biproducts, such categories were proposed as the framework for categorical quantum mechanics in \cite{Abramsky-Coecke}. The biproducts were eliminated using \emph{classical structures} in \cite{PavlovicD:QMWS}\footnote{We first called them classical \emph{objects}, but too many people pointed out that they had a structure. Although one of the main points of category theory is to make structures into objects (e.g. groups have a structure, but they are objects of the category of groups), it seemed simpler to change the name than to explain one of the main points of category theory.}. The availability of classical structures over the objects of a dagger-compact category is analogous to the availability of bases in the category of Hilbert spaces. Instantiated to this category, classical structures \cite[Def.~2.2]{PavlovicD:QMWS,PavlovicD:CQStruct} in fact exactly correspond to bases \cite{PavlovicD:MSCS13}. Although classical structures are generally not preserved by the morphisms of the surrounding dagger-compact category (just like the bases are not preserved by linear operators), they do influence the compact structure, by providing an isomorphism between each object and its dual, and thus allow us to choose the dual to be $X^\ast = X$, and thus make each object self-dual \cite[Prop.~2.4]{PavlovicD:CQStruct}. The \emph{Frobenius condition} imposed on adjoint monoid-comonoid pairs  \cite{Carboni-Walters} is just another way to express this self-duality \cite[Thm.~4.3]{PavlovicD:Qabs12}. Yet another expression of the same is an \emph{entangled}\/ vector $I\tto\eta X\otimes X$, i.e. such that $(\eta^\ddag \otimes X)\circ(X\otimes \eta) = \id$ \cite[Prop.~2.6]{PavlovicD:Qabs12}. We use such vectors below. Dagger-compact categories with such self-dualities, or classical structures, playing the role of bases to capture classical data, were studied as  \emph{categories of classical structures} in \cite[Sec.~2.2]{PavlovicD:CQStruct}.

\subsection{Orthocomplement}
Let $A$ be an object in a dagger-compact category $\CCc$, given with a classical structure induced by the monoid
\[A \otimes A \tto{\ \mnd\ } A \oot{\ \cun\ } I \]
Suppose that, in addition to this classical monoid, we are also given another commutative monoid
\[A \otimes A \tto{\ \ovee\ } A \oot{\ 0\ } I \]

\begin{definition} An \emph{orthocomplement} with respect to the commutative monoid $(A,\ovee, 0)$ is an operation $\ortt: A\to A$ such that the equations
\newpage
\beq\label{eq:complement}
 \def\JPicScale{.6}\newcommand{\oovee}{\scriptstyle \ovee}\newcommand{\aah}{\scriptstyle A}\newcommand{\oone}{\scriptscriptstyle \iota}\newcommand{\EQLS}{=} \newcommand{\ahh}{\scriptstyle \ortt}
\ifx\JPicScale\undefined\def\JPicScale{1}\fi
\psset{unit=\JPicScale mm}
\psset{linewidth=0.3,dotsep=1,hatchwidth=0.3,hatchsep=1.5,shadowsize=1,dimen=middle}
\psset{dotsize=0.7 2.5,dotscale=1 1,fillcolor=black}
\psset{arrowsize=1 2,arrowlength=1,arrowinset=0.25,tbarsize=0.7 5,bracketlength=0.15,rbracketlength=0.15}
\begin{pspicture}(0,0)(100,25.62)
\psline[linewidth=0.35,border=1.05](18.12,3.75)
(18.12,-5)(25,-9.38)
\psline[linewidth=0.35](31.88,22.5)
(31.88,-5)(25,-9.38)
\psline[linewidth=0.35](5.62,-25)(5.62,3.75)
\pscustom[linewidth=0.35]{\psline(23.75,3.74)(0,3.75)
\psline(0,3.75)(8.12,11.87)
\psbezier(8.12,11.87)(8.12,11.87)(8.12,11.87)
\psline(8.12,11.87)(15.62,11.88)
\psline(15.62,11.88)(23.75,3.74)
\psbezier(23.75,3.74)(23.75,3.74)(23.75,3.74)
\psbezier(23.75,3.74)(23.75,3.74)(23.75,3.74)
}
\psline[linewidth=0.35](11.88,11.88)(11.88,14.38)
\rput(11.88,7.5){$\oovee$}
\pscustom[linewidth=0.35]{\psline(15.62,14.37)(8.12,14.37)
\psline(8.12,14.37)(11.87,18.12)
\psbezier(11.87,18.12)(11.87,18.12)(11.87,18.12)
\psline(11.87,18.12)(15.62,14.37)
\psline(15.62,14.37)(15,14.38)
\psbezier(15,14.38)(15,14.38)(15,14.38)
\psbezier(15,14.38)(15,14.38)(15,14.38)
}
\rput(11.88,16.25){$\oone$}
\psline(95,14.25)(95,23.25)
\rput(95,9.25){$\ahh$}
\psline(90,9.25)
(90,4.25)
(100,4.25)
(100,14.25)
(90,14.25)
(100,14.25)
(90,14.25)(90,9.25)
\psline(95,-25.62)(95,-14.5)
\psline(95,-4.5)(95,4.5)
\rput(95,-9.5){$\ahh$}
\rput(95,-28.12){$\aah$}
\psline(90,-9.5)
(90,-14.5)
(100,-14.5)
(100,-4.5)
(90,-4.5)
(100,-4.5)
(90,-4.5)(90,-9.5)
\psline(60,-25.62)(60,22.5)
\rput(60,-28.75){$\aah$}
\rput(75,0){$\EQLS$}
\rput(45.62,0){$\EQLS$}
\psline[fillcolor=white,fillstyle=solid](0.62,-10.12)
(0.62,-15.12)
(10.62,-15.12)
(10.62,-5.12)
(0.62,-5.12)
(10.62,-5.12)
(0.62,-5.12)(0.62,-10.12)
\rput(5.62,-28.12){$\aah$}
\rput(95,25.62){$\aah$}
\rput(31.88,25.62){$\aah$}
\rput{0}(25,-8.75){\psellipse[linewidth=0.35,fillstyle=solid](0,0)(1.48,-1.48)}
\rput{0}(25,-17.5){\psellipse[linewidth=0.35,fillstyle=solid](0,0)(1.48,-1.48)}
\pscustom[linewidth=0.35]{\psline(13.13,-4.99)(36.88,-5)
\psline(36.88,-5)(28.76,-13.12)
\psbezier(28.76,-13.12)(28.76,-13.12)(28.76,-13.12)
\psline(28.76,-13.12)(21.26,-13.13)
\psline(21.26,-13.13)(13.13,-4.99)
\psbezier(13.13,-4.99)(13.13,-4.99)(13.13,-4.99)
\psbezier(13.13,-4.99)(13.13,-4.99)(13.13,-4.99)
}
\psline[linewidth=0.35](25,-8.75)(25,-17.5)
\pscustom[linewidth=0.35]{\psline(20,-15.62)(30,-15.62)
\psline(30,-15.62)(25,-20.62)
\psbezier(25,-20.62)(25,-20.62)(25,-20.62)
\psline(25,-20.62)(20,-15.62)
\psline(20,-15.62)(20.62,-15.62)
\psbezier(20.62,-15.62)(20.62,-15.62)(20.62,-15.62)
\psbezier(20.62,-15.62)(20.62,-15.62)(20.62,-15.62)
}
\rput(5.62,-10.62){$\ahh$}
\end{pspicture}

\eeq

\vspace{3.5\baselineskip}
\noindent hold for some $\iota \in A$.
\end{definition}

\paragraph{Remark.} These equations can be construed as a string diagrammatic version of the equations
\beq\label{eq:compl}  x \ovee \ort x = \iota \qquad \qquad \qquad \ort{\ort x} = x \eeq
However, the formal correspondence between of the left-hand equation in \eqref{eq:complement} and the left-hand equation in \eqref{eq:compl} depends on the single-valuedness assumption, which will be discussed in the next section.

It turns out that the orthocomplement operations over a monoid are in bijective correspondence with the \emph{unbiased}\/ vectors with respect to it. We first define and then explain what this means.

\begin{definition}\label{def:unbiased} An element $\iota\in A$ is said to be \emph{unbiased} with respect to the commutative monoid structure $(A, \ovee, 0)$ if it satisfies the equation
\beq\label{eq:unbiased}
 \def\JPicScale{.6}\newcommand{\oovee}{\scriptstyle \ovee}\newcommand{\aah}{\scriptstyle A}\newcommand{\oone}{\scriptscriptstyle \iota}\newcommand{\EQLS}{=} \newcommand{\ahh}{\scriptstyle \ortt}
\ifx\JPicScale\undefined\def\JPicScale{1}\fi
\psset{unit=\JPicScale mm}
\psset{linewidth=0.3,dotsep=1,hatchwidth=0.3,hatchsep=1.5,shadowsize=1,dimen=middle}
\psset{dotsize=0.7 2.5,dotscale=1 1,fillcolor=black}
\psset{arrowsize=1 2,arrowlength=1,arrowinset=0.25,tbarsize=0.7 5,bracketlength=0.15,rbracketlength=0.15}
\begin{pspicture}(0,0)(68.12,25)
\psline[linewidth=0.35](5.62,-21.25)(5.62,2.5)
\pscustom[linewidth=0.35]{\psline(23.75,2.49)(0,2.5)
\psline(0,2.5)(8.12,10.62)
\psbezier(8.12,10.62)(8.12,10.62)(8.12,10.62)
\psline(8.12,10.62)(15.62,10.63)
\psline(15.62,10.63)(23.75,2.49)
\psbezier(23.75,2.49)(23.75,2.49)(23.75,2.49)
\psbezier(23.75,2.49)(23.75,2.49)(23.75,2.49)
}
\psline[linewidth=0.35](11.88,10.63)(11.88,13.13)
\rput(11.88,6.25){$\oovee$}
\pscustom[linewidth=0.35]{\psline(15.62,13.12)(8.12,13.12)
\psline(8.12,13.12)(11.87,16.87)
\psbezier(11.87,16.87)(11.87,16.87)(11.87,16.87)
\psline(11.87,16.87)(15.62,13.12)
\psline(15.62,13.12)(15,13.13)
\psbezier(15,13.13)(15,13.13)(15,13.13)
\psbezier(15,13.13)(15,13.13)(15,13.13)
}
\rput(11.88,15){$\oone$}
\psline(68.12,-22.5)(68.12,23.75)
\rput(68.12,-25){$\aah$}
\rput(50.62,0){$\EQLS$}
\rput(5,-25){$\aah$}
\psline[linewidth=0.35](32.5,21.25)(32.5,-2.5)
\pscustom[linewidth=0.35]{\psline(14.38,-2.49)(38.12,-2.5)
\psline(38.12,-2.5)(30.01,-10.62)
\psbezier(30.01,-10.62)(30.01,-10.62)(30.01,-10.62)
\psline(30.01,-10.62)(22.51,-10.63)
\psline(22.51,-10.63)(14.38,-2.49)
\psbezier(14.38,-2.49)(14.38,-2.49)(14.38,-2.49)
\psbezier(14.38,-2.49)(14.38,-2.49)(14.38,-2.49)
}
\psline[linewidth=0.35](26.25,-10.63)(26.25,-13.13)
\rput(26.25,-6.25){$\oovee$}
\pscustom[linewidth=0.35]{\psline(22.51,-13.12)(30.01,-13.12)
\psline(30.01,-13.12)(26.26,-16.87)
\psbezier(26.26,-16.87)(26.26,-16.87)(26.26,-16.87)
\psline(26.26,-16.87)(22.51,-13.12)
\psline(22.51,-13.12)(23.13,-13.13)
\psbezier(23.13,-13.13)(23.13,-13.13)(23.13,-13.13)
\psbezier(23.13,-13.13)(23.13,-13.13)(23.13,-13.13)
}
\rput(26.25,-15){$\oone$}
\rput(33.12,25){$\aah$}
\psline(19.38,-2.5)(19.38,2.5)
\end{pspicture}

\eeq
\end{definition}
\bigskip
\bigskip
\medskip

\paragraph{Explanation.} In the terminology of \cite[Definitions~2.5 and 5.1]{PavlovicD:Qabs12}, a vector $I\tto \iota A$ is unbiased with respect to an algebra with the underlying monoid $(A, \ovee, 0)$ in a dagger-compact category just when the vector $I \tto \iota A \tto{\ovee^\ddag} A \otimes A$ is \emph{entangled}; and the entanglement is defined by the equation \eqref{eq:unbiased}. Entangled vectors are often also called \emph{Bell states} \cite[Sec.~2.1]{PavlovicD:CQStruct}.  Intuitively, a vector $I\tto \varphi A\otimes A$ is entangled if it implements an inner product $<a|b> = \varphi^\ddag \circ(a_\ast \otimes b)$ \cite[Prop.~2.6]{PavlovicD:Qabs12}, which means that the induced linear operator $A\tto{\widehat \varphi} A$ is unitary \cite[Prop.~5.2(a)]{PavlovicD:Qabs12}. Def.~\ref{def:unbiased} is also equivalent to \cite[Def.~7.13]{Coecke-Duncan:ZX} up to a scalar. 

\begin{proposition}\label{prop:ortho}
The orthocomplement operations $A\tto\ortt A$ with respect to a commutative monoid $(A,\ovee,0)$ are in a bijective correspondence with its unbiased vectors $I\tto \iota A$. 
%
%
\end{proposition}

\begin{proof}
Given an orthocomplement, conditions \eqref{eq:complement} immediately imply
\[
 \def\JPicScale{.6}\newcommand{\oovee}{\scriptstyle \ovee}\newcommand{\zzero}{\scriptscriptstyle 0}\newcommand{\aah}{\scriptstyle A}\newcommand{\oone}{\scriptscriptstyle \iota}\newcommand{\EQLS}{=}\newcommand{\ANDD}{and} \newcommand{\ahh}{\scriptstyle \ortt}
\ifx\JPicScale\undefined\def\JPicScale{1}\fi
\psset{unit=\JPicScale mm}
\psset{linewidth=0.3,dotsep=1,hatchwidth=0.3,hatchsep=1.5,shadowsize=1,dimen=middle}
\psset{dotsize=0.7 2.5,dotscale=1 1,fillcolor=black}
\psset{arrowsize=1 2,arrowlength=1,arrowinset=0.25,tbarsize=0.7 5,bracketlength=0.15,rbracketlength=0.15}
\begin{pspicture}(0,0)(156.88,21.24)
\psline(5,-11.25)(5,-3.25)
\psline(5,6.75)(5,16.25)
\rput(5,1.75){$\ahh$}
\rput(5,-14.38){$\aah$}
\psline(0,1.75)
(0,-3.25)
(10,-3.25)
(10,6.75)
(0,6.75)
(10,6.75)
(0,6.75)(0,1.75)
\psline[linewidth=0.35,border=1.05](47.5,6.88)
(47.5,-1.88)(54.38,-6.26)
\psline[linewidth=0.35](61.25,15.62)
(61.25,-1.88)(54.38,-6.26)
\psline[linewidth=0.35](35,-11.25)(35,6.88)
\pscustom[linewidth=0.35]{\psline(53.12,6.87)(29.38,6.88)
\psline(29.38,6.88)(37.5,15)
\psbezier(37.5,15)(37.5,15)(37.5,15)
\psline(37.5,15)(45,15)
\psline(45,15)(53.12,6.87)
\psbezier(53.12,6.87)(53.12,6.87)(53.12,6.87)
\psbezier(53.12,6.87)(53.12,6.87)(53.12,6.87)
}
\psline[linewidth=0.35](41.25,15)(41.25,17.5)
\rput(41.25,10.62){$\oovee$}
\pscustom[linewidth=0.35]{\psline(45,17.5)(37.5,17.5)
\psline(37.5,17.5)(41.25,21.24)
\psbezier(41.25,21.24)(41.25,21.24)(41.25,21.24)
\psline(41.25,21.24)(45,17.5)
\psline(45,17.5)(44.38,17.5)
\psbezier(44.38,17.5)(44.38,17.5)(44.38,17.5)
\psbezier(44.38,17.5)(44.38,17.5)(44.38,17.5)
}
\rput(41.25,19.38){$\oone$}
\rput(20.62,1.88){$\EQLS$}
\rput(35,-14.38){$\aah$}
\rput(61.25,19.38){$\aah$}
\rput(91.25,1.88){\ANDD}
\psline[linewidth=0.35](124.38,-12.5)(124.38,15.62)
\rput(124.38,19.38){$\aah$}
\pscustom[linewidth=0.35]{\psline(119.38,-12.5)(129.38,-12.5)
\psline(129.38,-12.5)(124.38,-17.49)
\psbezier(124.38,-17.49)(124.38,-17.49)(124.38,-17.49)
\psline(124.38,-17.49)(119.38,-12.5)
\psline(119.38,-12.5)(121.26,-12.5)
\psbezier(121.26,-12.5)(121.26,-12.5)(121.26,-12.5)
\psbezier(121.26,-12.5)(121.26,-12.5)(121.26,-12.5)
}
\rput(124.38,-14.38){$\oone$}
\rput(5,19.38){$\aah$}
\psline(151.88,6.76)(151.88,15.62)
\rput(151.88,1.76){$\ahh$}
\psline(146.88,1.76)
(146.88,-3.24)
(156.88,-3.24)
(156.88,6.76)
(146.88,6.76)
(156.88,6.76)
(146.88,6.76)(146.88,1.76)
\rput(151.25,19.38){$\aah$}
\rput(136.25,2.5){$\EQLS$}
\rput(151.88,-14.38){$\zzero$}
\psline[linewidth=0.35](151.88,-12.5)(151.88,-3.12)
\pscustom[linewidth=0.35]{\psline(146.88,-12.5)(156.88,-12.5)
\psline(156.88,-12.5)(151.88,-17.49)
\psbezier(151.88,-17.49)(151.88,-17.49)(151.88,-17.49)
\psline(151.88,-17.49)(146.88,-12.5)
\psline(146.88,-12.5)(148.75,-12.5)
\psbezier(148.75,-12.5)(148.75,-12.5)(148.75,-12.5)
\psbezier(148.75,-12.5)(148.75,-12.5)(148.75,-12.5)
}
\rput{0}(54.38,-5.62){\psellipse[linewidth=0.35,fillstyle=solid](0,0)(1.48,-1.48)}
\rput{0}(54.38,-14.38){\psellipse[linewidth=0.35,fillstyle=solid](0,0)(1.48,-1.48)}
\pscustom[linewidth=0.35]{\psline(42.5,-1.87)(66.25,-1.88)
\psline(66.25,-1.88)(58.13,-10)
\psbezier(58.13,-10)(58.13,-10)(58.13,-10)
\psline(58.13,-10)(50.63,-10.01)
\psline(50.63,-10.01)(42.5,-1.87)
\psbezier(42.5,-1.87)(42.5,-1.87)(42.5,-1.87)
\psbezier(42.5,-1.87)(42.5,-1.87)(42.5,-1.87)
}
\psline[linewidth=0.35](54.38,-5.62)(54.38,-14.38)
\pscustom[linewidth=0.35]{\psline(49.38,-12.5)(59.38,-12.5)
\psline(59.38,-12.5)(54.38,-17.5)
\psbezier(54.38,-17.5)(54.38,-17.5)(54.38,-17.5)
\psline(54.38,-17.5)(49.38,-12.5)
\psline(49.38,-12.5)(50,-12.5)
\psbezier(50,-12.5)(50,-12.5)(50,-12.5)
\psbezier(50,-12.5)(50,-12.5)(50,-12.5)
}
\end{pspicture}

\]
\bigskip

\noindent which shows that the orthocomplement $\ortt$ and the element $\iota$ uniquely determine each other. But if the orthocomplement satisfies the left hand equation, then it is easy to see that \eqref{eq:unbiased} holds if and only if $\ortt \ortt = \id$, as in \eqref{eq:complement}.
\end{proof}

\begin{definition} An \emph{orthocomplemented monoid} over a classical structure $A$ is a tuple $(A,\ovee,0,1,\ortt)$, where \begin{itemize}
\item $(A,\ovee,0)$ is a commutative monoid, 
\item $I\tto 1 A$ is an unbiased vector, and 
\item $A \tto \ortt A$ is the induced orthocomplementation.
\end{itemize}
\end{definition}

\begin{proposition}\label{prop:demorgan}
$(A,\ovee,0,1,\ortt)$ is an orthocomplemented monoid iff
$(A,\owedge,1,0,\ortt)$ is an orthocomplemented monoid, where
\vspace{-.5\baselineskip}
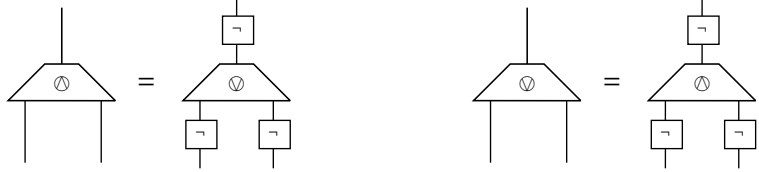
\[
 \def\JPicScale{.6}\newcommand{\oovee}{\scriptstyle \ovee}\newcommand{\oort}{\scriptscriptstyle \ortt}\newcommand{\oowedge}{\scriptstyle \owedge}\newcommand{\EQLS}{=}
\ifx\JPicScale\undefined\def\JPicScale{1}\fi
\psset{unit=\JPicScale mm}
\psset{linewidth=0.3,dotsep=1,hatchwidth=0.3,hatchsep=1.5,shadowsize=1,dimen=middle}
\psset{dotsize=0.7 2.5,dotscale=1 1,fillcolor=black}
\psset{arrowsize=1 2,arrowlength=1,arrowinset=0.25,tbarsize=0.7 5,bracketlength=0.15,rbracketlength=0.15}
\begin{pspicture}(0,0)(165.63,19.38)
\pscustom[linewidth=0.35]{\psline(165.63,-3.76)(141.87,-3.74)
\psline(141.87,-3.74)(149.99,4.38)
\psbezier(149.99,4.38)(149.99,4.38)(149.99,4.38)
\psline(149.99,4.38)(157.49,4.38)
\psline(157.49,4.38)(165.63,-3.76)
\psbezier(165.63,-3.76)(165.63,-3.76)(165.63,-3.76)
\psbezier(165.63,-3.76)(165.63,-3.76)(165.63,-3.76)
}
\rput(115,0){$\oovee$}
\rput(133.75,0){$\EQLS$}
\psline(145.62,-8.12)(145.63,-3.74)
\pscustom[]{\psline(142.5,-8.75)(142.5,-14.38)
\psline(142.5,-14.38)(149.38,-14.38)
\psline(149.38,-14.38)(149.38,-8.12)
\psline(149.38,-8.12)(142.5,-8.12)
\psline(142.5,-8.12)(142.5,-8.75)
\psbezier(142.5,-8.75)(142.5,-8.75)(142.5,-8.75)
\psbezier(142.5,-8.75)(142.5,-8.75)(142.5,-8.75)
}
\pscustom[linewidth=0.35]{\psline(126.88,-3.76)(103.12,-3.74)
\psline(103.12,-3.74)(111.24,4.38)
\psbezier(111.24,4.38)(111.24,4.38)(111.24,4.38)
\psline(111.24,4.38)(118.74,4.38)
\psline(118.74,4.38)(126.88,-3.76)
\psbezier(126.88,-3.76)(126.88,-3.76)(126.88,-3.76)
\psbezier(126.88,-3.76)(126.88,-3.76)(126.88,-3.76)
}
\psline[linewidth=0.35](115,4.38)(115,16.88)
\psline(106.88,-17.5)(106.88,-3.74)
\psline(123.75,-17.5)(123.75,-3.74)
\psline(145.62,-18.75)(145.63,-14.37)
\rput(145.62,-11.25){$\oort$}
\psline(161.88,-8.12)(161.88,-3.74)
\pscustom[]{\psline(158.75,-8.75)(158.75,-14.38)
\psline(158.75,-14.38)(165.62,-14.38)
\psline(165.62,-14.38)(165.62,-8.12)
\psline(165.62,-8.12)(158.75,-8.12)
\psline(158.75,-8.12)(158.75,-8.75)
\psbezier(158.75,-8.75)(158.75,-8.75)(158.75,-8.75)
\psbezier(158.75,-8.75)(158.75,-8.75)(158.75,-8.75)
}
\psline(161.88,-18.75)(161.88,-14.37)
\rput(161.88,-11.25){$\oort$}
\psline(153.75,15)(153.76,19.38)
\pscustom[]{\psline(150.62,14.38)(150.62,8.75)
\psline(150.62,8.75)(157.5,8.75)
\psline(157.5,8.75)(157.5,15)
\psline(157.5,15)(150.62,15)
\psline(150.62,15)(150.62,14.38)
\psbezier(150.62,14.38)(150.62,14.38)(150.62,14.38)
\psbezier(150.62,14.38)(150.62,14.38)(150.62,14.38)
}
\psline(153.75,4.38)(153.76,8.76)
\rput(153.75,11.88){$\oort$}
\rput(153.75,0){$\oowedge$}
\pscustom[linewidth=0.35]{\psline(62.5,-3.76)(38.75,-3.74)
\psline(38.75,-3.74)(46.87,4.38)
\psbezier(46.87,4.38)(46.87,4.38)(46.87,4.38)
\psline(46.87,4.38)(54.37,4.38)
\psline(54.37,4.38)(62.5,-3.76)
\psbezier(62.5,-3.76)(62.5,-3.76)(62.5,-3.76)
\psbezier(62.5,-3.76)(62.5,-3.76)(62.5,-3.76)
}
\rput(50.62,0){$\oovee$}
\rput(30.62,0){$\EQLS$}
\psline(42.5,-8.12)(42.5,-3.74)
\pscustom[]{\psline(39.38,-8.75)(39.38,-14.38)
\psline(39.38,-14.38)(46.25,-14.38)
\psline(46.25,-14.38)(46.25,-8.12)
\psline(46.25,-8.12)(39.38,-8.12)
\psline(39.38,-8.12)(39.38,-8.75)
\psbezier(39.38,-8.75)(39.38,-8.75)(39.38,-8.75)
\psbezier(39.38,-8.75)(39.38,-8.75)(39.38,-8.75)
}
\pscustom[linewidth=0.35]{\psline(23.75,-3.76)(-0,-3.74)
\psline(-0,-3.74)(8.11,4.38)
\psbezier(8.11,4.38)(8.11,4.38)(8.11,4.38)
\psline(8.11,4.38)(15.62,4.38)
\psline(15.62,4.38)(23.75,-3.76)
\psbezier(23.75,-3.76)(23.75,-3.76)(23.75,-3.76)
\psbezier(23.75,-3.76)(23.75,-3.76)(23.75,-3.76)
}
\psline[linewidth=0.35](11.88,4.38)(11.88,16.88)
\psline(3.75,-17.5)(3.75,-3.74)
\psline(20.62,-17.5)(20.62,-3.74)
\psline(42.5,-18.75)(42.5,-14.37)
\rput(42.5,-11.25){$\oort$}
\psline(58.75,-8.12)(58.75,-3.74)
\pscustom[]{\psline(55.62,-8.75)(55.62,-14.38)
\psline(55.62,-14.38)(62.5,-14.38)
\psline(62.5,-14.38)(62.5,-8.12)
\psline(62.5,-8.12)(55.62,-8.12)
\psline(55.62,-8.12)(55.62,-8.75)
\psbezier(55.62,-8.75)(55.62,-8.75)(55.62,-8.75)
\psbezier(55.62,-8.75)(55.62,-8.75)(55.62,-8.75)
}
\psline(58.75,-18.75)(58.75,-14.37)
\rput(58.75,-11.25){$\oort$}
\psline(50.62,15)(50.63,19.38)
\pscustom[]{\psline(47.5,14.38)(47.5,8.75)
\psline(47.5,8.75)(54.38,8.75)
\psline(54.38,8.75)(54.38,15)
\psline(54.38,15)(47.5,15)
\psline(47.5,15)(47.5,14.38)
\psbezier(47.5,14.38)(47.5,14.38)(47.5,14.38)
\psbezier(47.5,14.38)(47.5,14.38)(47.5,14.38)
}
\psline(50.62,4.38)(50.63,8.76)
\rput(50.62,11.88){$\oort$}
\rput(11.88,0){$\oowedge$}
\end{pspicture}

 \]
 \vspace{.75\baselineskip}
\end{proposition}
 
 \begin{definition} An \emph{orthocomplemented algebra} over a classical object $A$ is the structure $(A,\ovee,\owedge,0,1,\ortt)$, where $(A,\ovee,0,1,\ortt)$ and $(A,\owedge,1,0,\ortt)$ are orthocomplemented monoids related by De Morgan's laws as in Prop.~\ref{prop:demorgan}.
\end{definition}

\paragraph{Comment.} On one hand, orthocomplemented algebras can be thought of as a generalization of Boolean algebras, which also have involutive negation and satisfy De Morgan's laws, and are indeed a special case. But on the other hand, they are a very special case, as some of the main features of Boolean algebras do not survive in orthocomplemented algebras, and make room for the main features of effect algebras. An orthocomplemented algebra structure is derived over an arbitrary commutative monoid $(A,\ovee,0)$ from an arbitrary unbiased element $\iota \in A$, which becomes $1$, and determines $\ortt$ and $\owedge$. The monoid is thus not extended by any new elements, but the structure of orthocomplemented algebra is derived from the monoid as it is --- by the magic of the entanglement engendered from  the unbiased element.

Truth be told, though, the monoid $(A,\ovee,0)$ cannot be completely arbitrary without causing degeneracies. For instance, if we take $(A,\ovee,0)$ to be a classical monoid $(A,\mnd,\cun)$, giving rise to a special commutative Frobenius algebra, then the induced orthocomplement $\ortt$ boils down to the identity and the whole structure collapses to $\owedge = \mnd = \ovee$, with $x\ovee \ort x = \ort x = x$. Many other monoids $(A,\ovee,0)$, different from the classical ones, also cause degeneracies. To avoid that, we must impose some \emph{special}\/ requirements, and some \emph{antispecial}\/ requirements.

\section{Special, antispecial and superspecial algebras}\label{Sec:Anti-effect}
\subsection{Convolution}\label{Sec:Convol}
Every internal monoid $B\otimes B \tto\mu B\oot\iota I$ in a monoidal category $\CCc$ induces an external monoid on the vectors (states) of type $B$, with the same unit, and
\bea\label{eq:ext}
\conv{}\mu\ :\  \CCc(I,B) \times \CCc(I,B) & \to & \CCc(I,B)\\
<x,y> & \mapsto & \mu\circ(x\otimes y)
\eea
Dually, any internal comonoid $A\otimes A\oot \lambda  A \tto \epsilon I$ induces an external monoid on the covectors (effects) of type $A$, with the same counit and
\bea
\conv\lambda {}\ :\  \CCc(A,I) \times \CCc(A,I) & \to & \CCc(A,I)\\
<u,v> & \mapsto & (u\otimes v)\circ \lambda 
\eea
Putting the two together, any comonoid-monoid pair $\left<I\oot \epsilon A\tto\lambda  A\otimes A,\ B\otimes B\tto\mu B \oot \iota B \right>$ induces a \emph{convolution monoid}
\bea\label{eq:conv}
\conv{\lambda }\mu\ :\  \CCc(A,B) \times \CCc(A,B) & \to & \CCc(A,B)\\
<f,g> & \mapsto & \mu\circ(f\otimes g)\circ \lambda 
\eea
with the unit $A\tto\epsilon I\tto \iota B$.

\begin{definition}\label{def:conv}
A\/ \emph{convolution algebra} in a monoidal category $\CCc$ is the tuple $(A,\mu,\iota,\lambda ,\epsilon)$, where 
$(A,\mu,\iota)$ is a commutative\footnote{The commutativity requirement is usually not imposed on convolutions. Here we only work with commutative monoids and comonoids, so we restrict the usual definition of convolution to avoid repeating the requirement.} monoid and $(A,\lambda ,\epsilon)$ is a commutative comonoid. 

A \emph{convolution monoid} $\convv\ =\ \  \conv \mu\lambda \  :\  \CCc(A,A)\times \CCc(A,A)\to \CCc(A,A)$ is induced by a convolution algebra as in \eqref{eq:conv}, or as in the following string diagram 
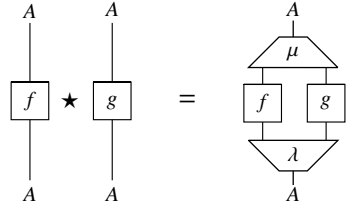
\[
 \def\JPicScale{.5}\newcommand{\oovee}{\scriptstyle \mu}\newcommand{\comonoid}{\scriptstyle \lambda}\newcommand{\aah}{\scriptstyle A}\newcommand{\bbh}{\scriptstyle A}\newcommand{\ahh}{\scriptstyle f}\newcommand{\bhh}{\scriptstyle g}\newcommand{\EQLS}{=}\newcommand{\Conv}{\convv} 
\ifx\JPicScale\undefined\def\JPicScale{1}\fi
\psset{unit=\JPicScale mm}
\psset{linewidth=0.3,dotsep=1,hatchwidth=0.3,hatchsep=1.5,shadowsize=1,dimen=middle}
\psset{dotsize=0.7 2.5,dotscale=1 1,fillcolor=black}
\psset{arrowsize=1 2,arrowlength=1,arrowinset=0.25,tbarsize=0.7 5,bracketlength=0.15,rbracketlength=0.15}
\begin{pspicture}(0,0)(88.75,27.5)
\psline(5,-18.12)(5,-2)
\psline(5,8)(5,23.75)
\rput(5,3){$\ahh$}
\rput(5,-21.88){$\aah$}
\psline(0,3)
(0,-2)
(10,-2)
(10,8)
(0,8)
(10,8)
(0,8)(0,3)
\pscustom[linewidth=0.35]{\psline(86.88,11.87)(63.12,11.88)
\psline(63.12,11.88)(71.24,20)
\psbezier(71.24,20)(71.24,20)(71.24,20)
\psline(71.24,20)(78.74,20)
\psline(78.74,20)(86.88,11.87)
\psbezier(86.88,11.87)(86.88,11.87)(86.88,11.87)
\psbezier(86.88,11.87)(86.88,11.87)(86.88,11.87)
}
\psline[linewidth=0.35](75,20)(75,24.38)
\rput(75,15.62){$\oovee$}
\rput(46.88,3.12){$\EQLS$}
\rput(75,-21.88){$\aah$}
\psline(26.88,-18.12)(26.88,-2)
\psline(26.88,8)(26.88,23.75)
\rput(26.88,-21.88){$\aah$}
\psline(21.88,3)
(21.88,-2)
(31.88,-2)
(31.88,8)
(21.88,8)
(31.88,8)
(21.88,8)(21.88,3)
\rput(15.62,3.12){$\Conv$}
\rput(26.88,2.5){$\bhh$}
\rput(5,26.88){$\bbh$}
\rput(26.88,27.5){$\bbh$}
\psline(66.88,-7.5)(66.88,-2.62)
\psline(66.88,7.38)(66.88,11.88)
\rput(66.88,2.38){$\ahh$}
\psline(61.88,2.38)
(61.88,-2.62)
(71.88,-2.62)
(71.88,7.38)
(61.88,7.38)
(71.88,7.38)
(61.88,7.38)(61.88,2.38)
\psline(83.75,-7.5)(83.75,-2.62)
\psline(83.75,7.38)(83.75,11.88)
\psline(78.75,2.38)
(78.75,-2.62)
(88.75,-2.62)
(88.75,7.38)
(78.75,7.38)
(88.75,7.38)
(78.75,7.38)(78.75,2.38)
\rput(83.75,2.5){$\bhh$}
\pscustom[linewidth=0.35]{\psline(63.12,-7.49)(86.88,-7.5)
\psline(86.88,-7.5)(78.76,-15.62)
\psbezier(78.76,-15.62)(78.76,-15.62)(78.76,-15.62)
\psline(78.76,-15.62)(71.26,-15.63)
\psline(71.26,-15.63)(63.12,-7.49)
\psbezier(63.12,-7.49)(63.12,-7.49)(63.12,-7.49)
\psbezier(63.12,-7.49)(63.12,-7.49)(63.12,-7.49)
}
\psline[linewidth=0.35](75,-15.63)(75,-19.38)
\rput(75,27.5){$\bbh$}
\rput(75,-11.25){$\comonoid$}
\end{pspicture}

\]
\end{definition}

\vspace{3\baselineskip}
\begin{definition}\label{def:spec}
A\/ convolution algebra $(A,\mu,\iota,\lambda ,\epsilon)$ is called
\begin{enumerate}[i.]
\item \emph{special} if $\id \convv \id$ is unitary, 
and 
\item \emph{antispecial} if $\id \convv \id$  is a scaled projector.
\end{enumerate}
\end{definition}

\paragraph{Remarks.} Recall that $e\in \CCc(A,A)$ is a 
\begin{enumerate}[i.]
\item \emph{unitary}\/ when $e\circ e^\ddag = e^\ddag\circ e = \id$; 
\item \emph{scaled projector} when $e = a\circ b^\ddag$ for some $a,b \in \CCc(I,A)$.
\end{enumerate}
In addition to \eqref{eq:ext}, any internal monoid $(B,\mu, \iota)$ also induces the \emph{Cayley}\/ representation
\bear
\Upsilon :\  \CCc(B) & \to & \CCc(B,B)\\
b  & \mapsto &  \mu\circ(b\otimes B)
\eear
When this monoid is a part of a classical structure, then with respect to this structure, the vector $b$ is
\begin{enumerate}[i.]
\item \emph{unbiased} if and only if $\Upsilon b$ is a unitary, and
\item a \emph{basis} vector if and only if $\Upsilon b$ is a pure projector.
\end{enumerate}
This is spelled out in \cite[Prop.~5.2]{Coecke-Duncan:ZX,PavlovicD:Qabs12}

\paragraph{Examples.} Every classical structure $(A,\mnd,\unt,\cmn,\cun)$ induces a convolution algebra \cite{PavlovicD:CQStruct}. When $\CCc = \FHilb$, then classical structures correspond to bases \cite{PavlovicD:MSCS13}, which induce the representations of morphisms $f,g \in \FHilb(A,B)$ as matrices and $f\convv g = \left(f_{ij} \cdot g_{ij}\right)_{n\times m}$ is the entrywise multiplication of the matrix representations $f = \left(f_{ij}\right)_{n\times m}$ and $g = \left(g_{ij}\right)_{n\times m}$. When $\CCc = \Rel$, then classical structures are disjoint unions of abelian groups \cite{PavlovicD:QI09}. With the additive notation for these group structures, the convolution of relations is
\bear a\big( R \convv S\big) b & \iff & \exists uv\in A\ \exists xy\in B.\ u+v = a \ \wedge\  u R x\  \wedge\  vSy\  \wedge\  x+y = b \eear
The standard classical structures in $\Rel$ can be viewed as the disjoint unions of the trivial group $\ZZz_1$, and for these standard classical structures, the convolution boils down to the intersection, i.e. $R \convv S = R\cap S$. 

\paragraph{Remark.} If every object in a dagger-compact category $\CCc$ have classical structures (like, e.g., all vector spaces have bases), then the induced convolutions make all hom-sets $\CCc(A,B)$ into abelian groups. This does not make $\CCc$ into an abelian category, because these convolutions are generally not preserved under composition. E.g., the relations $P\, ; (R\cap S)$ and $(P\, ; R)\cap (P\, ; S)$ coincide only if the relation $P$ is single-valued, i.e. a partial map.

It turns out that effect algebras are defined in terms of partial functions with a good reason.

\subsection{Maps}\label{Sec:Maps}
\begin{definition}\label{def:conv-preord}
The \emph{convolution preorder} induced by $\convv:\  \CCc(A,B) \times \CCc(A,B)  \to  \CCc(A,B)$ is the transitive reflexive relation $\leq$ on $\CCc(A,B)$ defined by
\bear
f \leq g & \iff & \exists \ell \in \CCc(A,B).\ f\convv \ell = g 
\eear
\end{definition}

\begin{definition}\label{def:tot-sv-map}
A morphism $f\in \CCc(A,B)$ in a dagger-compact category $\CCc$ is said to be
\begin{enumerate}[i.]
\item \emph{total}\ if $\id_A \leq f^\ddag \circ f$
\item \emph{single-valued}\ (or a \emph{partial map}) if $f\circ f^\ddag \leq \id_B$
\item a \emph{map}\/ if it is total and single-valued.
\end{enumerate}
The subcategories of $\CCc$ spanned by total, single-valued morphisms, and maps are denoted $\CCc_t$, $\CCc_s$ and $\CCc_m$, respectively.
\end{definition}

In a bicategory $\CCC$, a 1-cell $f\in \CCC(A,B)$ is called a map if it has a right adjoint $f^\ddag \in \CCC(B,A)$.  Remarkably, the maps within an arbitrary bicategory form an ordinary category. In particular, restricted to partial maps, the convolution preorder becomes a partial order, in the sense that $(f\leq g\ \wedge\  g\leq f)\  \Longrightarrow\  f=g$; and restricted to total maps, it becomes discrete, in the sense that $f\leq g \Longrightarrow f=g$. This remains true in a large family of bicategories \cite{PavlovicD:mapsI,PavlovicD:mapsII}. Here we do not need such results in full generality, but we will need the following lemma, instantiated to convolution preorders.

\begin{lemma}\label{lemma:sveq} For partial maps $f, g\in \CCc_{s}(A,B)$ the following holds
\beq\label{eq:conv-eq}
\newcommand{\eff}{\scriptstyle f}
\newcommand{\gee}{\scriptstyle g}
\newcommand{\geedag}{\scriptstyle g^\ddag}
\newcommand{\EQLS}{=}
\newcommand{\LEQ}{\leq}
\newcommand{\SAME}{\iff}
\newcommand{\ANDALSO}{and}
\def\JPicScale{.7}
\ifx\JPicScale\undefined\def\JPicScale{1}\fi
\psset{unit=\JPicScale mm}
\psset{linewidth=0.3,dotsep=1,hatchwidth=0.3,hatchsep=1.5,shadowsize=1,dimen=middle}
\psset{dotsize=0.7 2.5,dotscale=1 1,fillcolor=black}
\psset{arrowsize=1 2,arrowlength=1,arrowinset=0.25,tbarsize=0.7 5,bracketlength=0.15,rbracketlength=0.15}
\begin{pspicture}(0,0)(178.12,17.11)
\psline[linewidth=0.35](3.12,8.12)
(3.12,-8.13)(10,-15)
\psline[linewidth=0.35](16.88,8.12)
(16.88,-8.13)(10,-15)
\rput{90}(10,-15.62){\psellipse[linewidth=0.35,fillstyle=solid](0,0)(1.48,-1.48)}
\psline[linewidth=0.35]{->}(10,-23.12)(10,-15.63)
\rput(25,0){$\EQLS$}
\psline[linewidth=0.35](16.88,-7.5)
(16.87,8.13)(10,15)
\psline[linewidth=0.35](3.12,-8.12)
(3.12,8.13)(10,15)
\pscustom[linewidth=0.35,fillcolor=white,fillstyle=solid]{\psbezier(14.38,-2.5)(14.38,-2.5)(14.38,-2.5)(14.38,-2.5)
\psline(14.38,-2.5)(20,-2.5)
\psline(20,-2.5)(20,3.12)
\psline(20,3.12)(14.38,3.12)
\psline(14.38,3.12)(14.38,-2.5)
\psbezier(14.38,-2.5)(14.38,-2.5)(14.38,-2.5)
\psline(14.38,-2.5)(14.38,-1.25)
}
\pscustom[linewidth=0.35,fillcolor=white,fillstyle=solid]{\psbezier(0,-2.5)(0,-2.5)(0,-2.5)(0,-2.5)
\psline(0,-2.5)(5.62,-2.5)
\psline(5.62,-2.5)(5.62,3.12)
\psline(5.62,3.12)(0,3.12)
\psline(0,3.12)(0,-2.5)
\psbezier(0,-2.5)(0,-2.5)(0,-2.5)
\psline(0,-2.5)(0,-1.25)
}
\psline[linewidth=0.35](33.75,8.12)
(33.75,-8.13)(40.62,-15)
\psline[linewidth=0.35](47.5,8.12)
(47.5,-8.13)(40.62,-15)
\rput{0}(40.62,-15.62){\psellipse[linewidth=0.35,fillstyle=solid](0,0)(1.48,-1.48)}
\psline[linewidth=0.35]{->}(40.62,-23.12)(40.62,-15.63)
\psline[linewidth=0.35](47.5,-7.5)
(47.5,8.13)(40.62,15)
\psline[linewidth=0.35](33.75,-8.12)
(33.75,8.13)(40.62,15)
\pscustom[linewidth=0.35,fillcolor=white,fillstyle=solid]{\psbezier(31.25,1.25)(31.25,1.25)(31.25,1.25)(31.25,1.25)
\psline(31.25,1.25)(36.87,1.25)
\psline(36.87,1.25)(36.87,6.87)
\psline(36.87,6.87)(31.25,6.87)
\psline(31.25,6.87)(31.25,1.25)
\psbezier(31.25,1.25)(31.25,1.25)(31.25,1.25)
\psline(31.25,1.25)(31.25,2.5)
}
\pscustom[linewidth=0.35,fillcolor=white,fillstyle=solid]{\psbezier(31.25,-6.87)(31.25,-6.87)(31.25,-6.87)(31.25,-6.87)
\psline(31.25,-6.87)(36.88,-6.87)
\psline(36.88,-6.87)(36.88,-1.25)
\psline(36.88,-1.25)(31.25,-1.25)
\psline(31.25,-1.25)(31.25,-6.87)
\psbezier(31.25,-6.87)(31.25,-6.87)(31.25,-6.87)
\psline(31.25,-6.87)(31.25,-5.62)
}
\rput(3.12,0){$\eff$}
\rput(16.88,0){$\gee$}
\rput(33.75,-3.75){$\eff$}
\rput(34.38,4.38){$\geedag$}
\psline[linewidth=0.35](93.75,8.12)
(93.75,-8.13)(100.62,-15)
\psline[linewidth=0.35](107.5,8.12)
(107.5,-8.13)(100.62,-15)
\rput{0}(100.62,-15.62){\psellipse[linewidth=0.35,fillstyle=solid](0,0)(1.48,-1.48)}
\psline[linewidth=0.35]{->}(100.62,-23.12)(100.62,-15.63)
\rput(115.62,0){$\EQLS$}
\psline[linewidth=0.35](107.5,-7.5)
(107.5,8.13)(100.62,15)
\psline[linewidth=0.35](93.75,-8.12)
(93.75,8.13)(100.62,15)
\pscustom[linewidth=0.35,fillcolor=white,fillstyle=solid]{\psbezier(105,-2.5)(105,-2.5)(105,-2.5)(105,-2.5)
\psline(105,-2.5)(110.62,-2.5)
\psline(110.62,-2.5)(110.62,3.12)
\psline(110.62,3.12)(105,3.12)
\psline(105,3.12)(105,-2.5)
\psbezier(105,-2.5)(105,-2.5)(105,-2.5)
\psline(105,-2.5)(105,-1.25)
}
\pscustom[linewidth=0.35,fillcolor=white,fillstyle=solid]{\psbezier(90.62,-2.5)(90.62,-2.5)(90.62,-2.5)(90.62,-2.5)
\psline(90.62,-2.5)(96.24,-2.5)
\psline(96.24,-2.5)(96.24,3.12)
\psline(96.24,3.12)(90.62,3.12)
\psline(90.62,3.12)(90.62,-2.5)
\psbezier(90.62,-2.5)(90.62,-2.5)(90.62,-2.5)
\psline(90.62,-2.5)(90.62,-1.25)
}
\psline[linewidth=0.35](123.12,-22.5)(123.12,16.25)
\pscustom[linewidth=0.35,fillcolor=white,fillstyle=solid]{\psbezier(120.62,-2.5)(120.62,-2.5)(120.62,-2.5)(120.62,-2.5)
\psline(120.62,-2.5)(126.24,-2.5)
\psline(126.24,-2.5)(126.24,3.12)
\psline(126.24,3.12)(120.62,3.12)
\psline(120.62,3.12)(120.62,-2.5)
\psbezier(120.62,-2.5)(120.62,-2.5)(120.62,-2.5)
\psline(120.62,-2.5)(120.62,-1.25)
}
\rput(93.75,0){$\eff$}
\rput(107.5,0){$\gee$}
\rput(123.12,0){$\eff$}
\psline[linewidth=0.35](161.25,-22.5)(161.25,16.25)
\pscustom[linewidth=0.35,fillcolor=white,fillstyle=solid]{\psbezier(158.75,-2.5)(158.75,-2.5)(158.75,-2.5)(158.75,-2.5)
\psline(158.75,-2.5)(164.37,-2.5)
\psline(164.37,-2.5)(164.37,3.12)
\psline(164.37,3.12)(158.75,3.12)
\psline(158.75,3.12)(158.75,-2.5)
\psbezier(158.75,-2.5)(158.75,-2.5)(158.75,-2.5)
\psline(158.75,-2.5)(158.75,-1.25)
}
\rput(161.25,0){$\eff$}
\psline[linewidth=0.35](175,-22.5)(175,16.25)
\pscustom[linewidth=0.35,fillcolor=white,fillstyle=solid]{\psbezier(172.5,-2.5)(172.5,-2.5)(172.5,-2.5)(172.5,-2.5)
\psline(172.5,-2.5)(178.12,-2.5)
\psline(178.12,-2.5)(178.12,3.12)
\psline(178.12,3.12)(172.5,3.12)
\psline(172.5,3.12)(172.5,-2.5)
\psbezier(172.5,-2.5)(172.5,-2.5)(172.5,-2.5)
\psline(172.5,-2.5)(172.5,-1.25)
}
\rput(175,0){$\gee$}
\rput(144.38,0){$\SAME$}
\rput{0}(123.12,15.62){\psellipse[linewidth=0.35,fillstyle=solid](0,0)(1.48,-1.48)}
\rput(70,0){$\ANDALSO$}
\rput(168.12,0){$\LEQ$}
\rput{0}(100.62,15){\psellipse[linewidth=0.35,fillstyle=solid](0,0)(1.48,-1.48)}
\rput{0}(40.62,15){\psellipse[linewidth=0.35,fillstyle=solid](0,0)(1.48,-1.48)}
\rput{0}(10,15){\psellipse[linewidth=0.35,fillstyle=solid](0,0)(1.48,-1.48)}
\end{pspicture}

\eeq
\vspace{2\baselineskip}
%
\end{lemma}

If a dagger-compact category $\CCc$ admits a classical structure on every object, a fixed family of chosen convolution preorders on all hom-sets give rise to a \emph{cartesian  bicategory}\/ \cite{Carboni-Walters}. The following proposition is proved in \cite[Thm.~1.6, Lemma~2.5]{Carboni-Walters}.

\begin{proposition}\label{prop:maps-char}
In the cartesian bicategory $\CCc$ induced by a dagger-compact category with fixed classical structures (and the induced convolution preorders), the following equivalences hold for every morphism $f\in \CCc(A,B)$
\begin{enumerate}[i.]
\item $f$ is total if and only if $\cun_B \circ f = \cun_A$;
\item $f$ is single-valued if and only if $\blacktriangle_B \circ f = (f\otimes f) \circ \blacktriangle_A$
\item $f$ is a map if and only if it is a comonoid homomorphism between the classical structures on $A$ and $B$.
\end{enumerate}
\end{proposition}

\subsection{Effect algebras are superspecial}\label{Sec:superspec}
One direction of the following lemma follows directly from the definition of single-valuedness. The other direction also requires the observation that every monoid and every comonoid must be total (because they has the unit, and the counit, respectively).

\begin{lemma}\label{lemma:sing-spec}
A commutative monoid $(A,\ovee,0)$ in a dagger-compact category $\CCc$ with classical structures is single-valued with respect to these structures if and only if the induced convolution algebra $(A,\ovee,0,\ovee^\ddag,0^\ddag)$ is special.
\end{lemma}

The \emph{specialty}\/ requirement from Def.~\ref{def:spec} thus lifts to general dagger-compact categories the set theoretic restriction of Def.~\ref{def:effect}, that effect algebra structure consists of partial maps, i.e. that it is single-valued. The \emph{antispecialty}\/ requirement, on the other hand, lifts the rest of Def.~\ref{def:effect} to dagger-compact categories\footnote{Other interesting instances of the special and the antispecial requirements have been considered in \cite{CoeckeB:arithmetic,AmarH:LICS15}.}.

\begin{definition}\label{def:eff}
An orthocomplemented algebra $(A,\ovee,\owedge,0,1,\ortt)$  in a dagger-compact category $\CCc$ is said to be \emph{superspecial}
if it satisfies the following conditions:
\begin{enumerate}[(a)]
\item the convolution algebra $(A,\ovee,0,\ovee^\ddag,0^\ddag)$ is special, (or equivalently,\\ the convolution algebra $(A,\owedge,1,\owedge^\ddag,1^\ddag)$ is special), and
\item the convolution algebra $(A,\ovee,0,\owedge^\ddag,1^\ddag)$ is antispecial.
\end{enumerate}
\end{definition}

\begin{definition}\label{def:effect-gen}
Let $\CCc$ be a dagger-compact category with classical structures, and $\CCc_{s}$ the subcategory of single-valued morphisms. A \emph{general effect algebra}\/ is a diagram \eqref{eq:one}  in $\CCc_{s}$, such that $(A,\ovee,0)$ is a commutative monoid, and the diagrams in figure \eqref{eq:pb-string} are pullbacks.
\end{definition}

\begin{proposition}\label{prop:eff}
An orthocomplemented algebra $(A,\ovee,\owedge,0,1,\ortt)$ in a dagger-compact category $\CCc$ is superspecial 
if and only if $(A,\ovee,0,1,\ortt)$ is a general effect algebra in $\CCc$.
\end{proposition}

\begin{proof}
Since the equivalence between the specialty and the single-valuedness is clear from Lemma~\ref{lemma:sing-spec}, the task boils down to proving the equivalence between the antispecialty and the pullback conditions from Sec.~\ref{Sec:effect}. In the context of sets and partial functions of Def.~\ref{def:effect}, this equivalence means that conditions (\ref{eq:two}-\ref{eq:three}) hold if and only if the equations $x\ovee y = u$ and $x\owedge y = v$ are satisfied only for $u=1$ and $v=0$. 

To prove this in the context of a dagger-compact category $\CCc$, first note that the square
\medskip
\beq\label{eq:pb-1}
\newcommand{\Carrier}{\scriptstyle A}
\newcommand{\Pairs}{\scriptstyle A\otimes A}
\newcommand{\oone}{\scriptstyle \ortt}
\newcommand{\ttwo}{\scriptstyle \ortt\otimes \ortt}
\newcommand{\uup}{\scriptstyle\ovee}
\newcommand{\ddown}{\scriptstyle\owedge}
\def\JPicScale{.6}
\ifx\JPicScale\undefined\def\JPicScale{1}\fi
\psset{unit=\JPicScale mm}
\psset{linewidth=0.3,dotsep=1,hatchwidth=0.3,hatchsep=1.5,shadowsize=1,dimen=middle}
\psset{dotsize=0.7 2.5,dotscale=1 1,fillcolor=black}
\psset{arrowsize=1 2,arrowlength=1,arrowinset=0.25,tbarsize=0.7 5,bracketlength=0.15,rbracketlength=0.15}
\begin{pspicture}(0,0)(35.62,31.25)
\pscustom[]{\psline{<-}(32.5,0)(13.12,0)
\psbezier(13.12,0)(13.12,0)(13.12,0)
\psbezier{-}(13.12,0)(13.12,0)(13.12,0)
}
\pscustom[]{\psline{<-}(35,5)(35,25)
\psbezier(35,25)(35,25)(35,25)
\psbezier{-}(35,25)(35,25)(35,25)
}
\pscustom[]{\psline{<-}(5,5)(5,25)
\psbezier(5,25)(5,25)(5,25)
\psbezier{-}(5,25)(5,25)(5,25)
}
\rput(35,0){$\Carrier$}
\rput(5,0){$\Pairs$}
\rput(35,30){$\Carrier$}
\rput[l](35.62,15){$\oone$}
\rput[r](4.38,15){$\ttwo$}
\rput[b](20,31.25){$\uup$}
\rput[t](20,-1.25){$\ddown$}
\rput(5,30){$\Pairs$}
\pscustom[]{\psline{<-}(32.5,30)(13.12,30)
\psbezier(13.12,30)(13.12,30)(13.12,30)
\psbezier{-}(13.12,30)(13.12,30)(13.12,30)
}
\pscustom[]{\psline(10.62,23.76)(10.62,26.26)
\psbezier(10.62,26.26)(10.62,26.26)(10.62,26.26)
\psbezier(10.62,26.26)(10.62,26.26)(10.62,26.26)
}
\pscustom[]{\psline(10.62,23.76)(8.12,23.76)
\psbezier(8.12,23.76)(8.12,23.76)(8.12,23.76)
\psbezier(8.12,23.76)(8.12,23.76)(8.12,23.76)
}
\end{pspicture}

\eeq
is a pullback. Composing the left-hand square of diagram \eqref{eq:pb} with this pullback, and using the commutativity of the monoids, we conclude that all of the following three squares are pullbacks if and only if any of them is a pullback.

\smallskip
\begin{center}
\newcommand{\Unit}{\scriptstyle I}
\newcommand{\Carrier}{\scriptstyle A}
\newcommand{\Pairs}{\scriptstyle A\otimes A}
\newcommand{\PPairs}{\scriptstyle A\otimes A\otimes A\otimes A}
\newcommand{\oone}{\scriptstyle 1}
\newcommand{\ttwo}{\scriptstyle <\id, \ortt>}
\newcommand{\uup}{\scriptstyle\cun}
\newcommand{\ddown}{\scriptstyle\ovee}
\newcommand{\ooone}{\scriptstyle 0}
\newcommand{\tttwo}{\scriptstyle <\id, \ortt>}
\newcommand{\dddown}{\scriptstyle\owedge}
\newcommand{\oooone}{\scriptstyle <1,0>}
\newcommand{\ttttwo}{\scriptstyle <\pi_0,\ortt,\pi_1, \ortt>}
\newcommand{\ddddown}{\scriptstyle \ovee\otimes \owedge}
\newcommand{\IFF}{\iff}
\def\JPicScale{.6}
\ifx\JPicScale\undefined\def\JPicScale{1}\fi
\psset{unit=\JPicScale mm}
\psset{linewidth=0.3,dotsep=1,hatchwidth=0.3,hatchsep=1.5,shadowsize=1,dimen=middle}
\psset{dotsize=0.7 2.5,dotscale=1 1,fillcolor=black}
\psset{arrowsize=1 2,arrowlength=1,arrowinset=0.25,tbarsize=0.7 5,bracketlength=0.15,rbracketlength=0.15}
\begin{pspicture}(0,0)(220.62,31.25)
\pscustom[]{\psline{<-}(32.5,0)(13.12,0)
\psbezier(13.12,0)(13.12,0)(13.12,0)
\psbezier{-}(13.12,0)(13.12,0)(13.12,0)
}
\pscustom[]{\psline{<-}(35,5)(35,25)
\psbezier(35,25)(35,25)(35,25)
\psbezier{-}(35,25)(35,25)(35,25)
}
\pscustom[]{\psline{<-}(5,5)(5,25)
\psbezier(5,25)(5,25)(5,25)
\psbezier{-}(5,25)(5,25)(5,25)
}
\rput(35,0){$\Carrier$}
\rput(5,0){$\Pairs$}
\rput(5,29.38){$\Carrier$}
\rput[l](35.62,15){$\oone$}
\rput[r](4.38,15){$\ttwo$}
\rput[b](20,31.25){$\uup$}
\rput[t](20,-1.25){$\ddown$}
\pscustom[]{\psline{<-}(117.5,0)(98.12,0)
\psbezier(98.12,0)(98.12,0)(98.12,0)
\psbezier{-}(98.12,0)(98.12,0)(98.12,0)
}
\pscustom[]{\psline{<-}(120,5)(120,25)
\psbezier(120,25)(120,25)(120,25)
\psbezier{-}(120,25)(120,25)(120,25)
}
\pscustom[]{\psline{<-}(90,5)(90,25)
\psbezier(90,25)(90,25)(90,25)
\psbezier{-}(90,25)(90,25)(90,25)
}
\rput(120,0){$\Carrier$}
\rput(90,0){$\Pairs$}
\rput(90,29.38){$\Carrier$}
\rput[b](105,31.25){$\uup$}
\pscustom[]{\psline{<-}(211.88,0)(197.5,0)
\psbezier(197.5,0)(197.5,0)(197.5,0)
\psbezier{-}(197.5,0)(197.5,0)(197.5,0)
}
\pscustom[]{\psline{<-}(220,5)(220,25)
\psbezier(220,25)(220,25)(220,25)
\psbezier{-}(220,25)(220,25)(220,25)
}
\pscustom[]{\psline{<-}(180,5)(180,25)
\psbezier(180,25)(180,25)(180,25)
\psbezier{-}(180,25)(180,25)(180,25)
}
\rput(220.62,0){$\Pairs$}
\rput[b](200.62,31.25){$\uup$}
\pscustom[]{\psline{<-}(217.5,30)(188.12,30)
\psbezier(188.12,30)(188.12,30)(188.12,30)
\psbezier{-}(188.12,30)(188.12,30)(188.12,30)
}
\rput(35,29.38){$\Unit$}
\rput(120,29.38){$\Unit$}
\rput(60.62,15.62){$\IFF$}
\rput(143.75,15.62){$\IFF$}
\pscustom[]{\psline(10.62,24.38)(10.62,26.88)
\psbezier(10.62,26.88)(10.62,26.88)(10.62,26.88)
\psbezier(10.62,26.88)(10.62,26.88)(10.62,26.88)
}
\pscustom[]{\psline(10.62,24.38)(8.12,24.38)
\psbezier(8.12,24.38)(8.12,24.38)(8.12,24.38)
\psbezier(8.12,24.38)(8.12,24.38)(8.12,24.38)
}
\pscustom[]{\psline(95,24.38)(95,26.88)
\psbezier(95,26.88)(95,26.88)(95,26.88)
\psbezier(95,26.88)(95,26.88)(95,26.88)
}
\pscustom[]{\psline(95,24.38)(92.5,24.38)
\psbezier(92.5,24.38)(92.5,24.38)(92.5,24.38)
\psbezier(92.5,24.38)(92.5,24.38)(92.5,24.38)
}
\pscustom[]{\psline(185.62,24.38)(185.62,26.88)
\psbezier(185.62,26.88)(185.62,26.88)(185.62,26.88)
\psbezier(185.62,26.88)(185.62,26.88)(185.62,26.88)
}
\pscustom[]{\psline(185.62,24.38)(183.12,24.38)
\psbezier(183.12,24.38)(183.12,24.38)(183.12,24.38)
\psbezier(183.12,24.38)(183.12,24.38)(183.12,24.38)
}
\pscustom[]{\psline{<-}(33.75,30)(8.75,30)
\psline(8.75,30)(11.25,30)
\psbezier{-}(11.25,30)(11.25,30)(11.25,30)
}
\pscustom[]{\psline{<-}(118.75,30)(93.75,30)
\psline(93.75,30)(96.25,30)
\psbezier{-}(96.25,30)(96.25,30)(96.25,30)
}
\rput(180,30){$\Pairs$}
\rput[l](121.25,15.62){$\ooone$}
\rput[r](89.38,15.62){$\tttwo$}
\rput[t](105.62,-1.25){$\dddown$}
\rput(179.38,0){$\PPairs$}
\rput(220,30){$\Unit$}
\rput[r](179.38,16.25){$\ttttwo$}
\rput[l](220.62,16.25){$\oooone$}
\rput[t](203.75,-1.25){$\ddddown$}
\end{pspicture}

\end{center}
\medskip
Towards the third square, we prove the first equation in the following diagram. 
\beq\label{eq:U-1}
\newcommand{\oovee}{\ovee}
\newcommand{\oowedge}{\owedge}
\newcommand{\scal}{\scriptstyle\omega}
\newcommand{\oone}{\scriptscriptstyle 1}
\newcommand{\EQLS}{=}
\newcommand{\zzero}{\scriptscriptstyle 0}
\def\JPicScale{.5}
\ifx\JPicScale\undefined\def\JPicScale{1}\fi
\psset{unit=\JPicScale mm}
\psset{linewidth=0.3,dotsep=1,hatchwidth=0.3,hatchsep=1.5,shadowsize=1,dimen=middle}
\psset{dotsize=0.7 2.5,dotscale=1 1,fillcolor=black}
\psset{arrowsize=1 2,arrowlength=1,arrowinset=0.25,tbarsize=0.7 5,bracketlength=0.15,rbracketlength=0.15}
\begin{pspicture}(0,0)(226.25,37.5)
\psline[linewidth=0.35,border=1.05](2.5,21.88)
(2.5,9.37)(10,1.87)
\psline[linewidth=0.35](16.25,21.88)
(16.25,9.37)(8.75,1.87)
\rput{0}(9.38,1.88){\psellipse[linewidth=0.35,fillstyle=solid](0,0)(1.48,-1.48)}
\psline[linewidth=0.2](-8.12,18.75)
(-8.12,-57.5)
(75.62,-57.5)
(75.62,18.75)
(-8.12,18.75)
(-8.12,8.75)
(-8.12,3.12)(-8.12,-9.38)
\pscustom[linewidth=0.35]{\psline(21.88,21.86)(-1.88,21.88)
\psline(-1.88,21.88)(6.25,30)
\psbezier(6.25,30)(6.25,30)(6.25,30)
\psline(6.25,30)(13.75,30)
\psline(13.75,30)(21.88,21.86)
\psbezier(21.88,21.86)(21.88,21.86)(21.88,21.86)
\psline(21.88,21.86)(21.87,21.86)
}
\psline[linewidth=0.35](110,-10)
(120,0)
(130,-10)
(120,-20)(110,-10)
\psline[linewidth=0.35](10,30)(10,37.5)
\pscustom[linewidth=0.35,fillcolor=white,fillstyle=solid]{\psbezier(13.75,11.25)(13.75,11.25)(13.75,11.25)(13.75,11.25)
\psline(13.75,11.25)(19.37,11.25)
\psline(19.37,11.25)(19.37,16.87)
\psline(19.37,16.87)(13.75,16.87)
\psline(13.75,16.87)(13.75,11.25)
\psbezier(13.75,11.25)(13.75,11.25)(13.75,11.25)
\psline(13.75,11.25)(13.75,12.5)
}
\rput(16.25,13.75){$\ortt$}
\psline[linewidth=0.35](112.49,10.01)(112.49,36.24)
\rput(90,0){$\EQLS$}
\rput(10,25.62){$\oovee$}
\rput(112.49,7.51){$\oone$}
\pscustom[linewidth=0.35]{\psline(118.13,10.01)(112.49,4.38)
\psline(112.49,4.38)(106.87,10.01)
\psline(106.87,10.01)(118.13,10.01)
\psbezier(118.13,10.01)(118.13,10.01)(118.13,10.01)
}
\psline[linewidth=0.35](50,21.88)
(50,9.38)(57.5,1.88)
\psline[linewidth=0.35](63.75,21.88)
(63.75,9.38)(56.25,1.88)
\rput{0}(56.88,1.88){\psellipse[linewidth=0.35,fillstyle=solid](0,0)(1.48,-1.48)}
\pscustom[linewidth=0.35]{\psline(69.38,21.86)(45.62,21.88)
\psline(45.62,21.88)(53.75,30)
\psbezier(53.75,30)(53.75,30)(53.75,30)
\psline(53.75,30)(61.25,30)
\psline(61.25,30)(69.38,21.86)
\psbezier(69.38,21.86)(69.38,21.86)(69.38,21.86)
\psline(69.38,21.86)(69.37,21.86)
}
\psline[linewidth=0.35](57.5,30)(57.5,37.5)
\pscustom[linewidth=0.35,fillcolor=white,fillstyle=solid]{\psbezier(61.25,11.24)(61.25,11.24)(61.25,11.24)(61.25,11.24)
\psline(61.25,11.24)(66.87,11.24)
\psline(66.87,11.24)(66.87,16.86)
\psline(66.87,16.86)(61.25,16.86)
\psline(61.25,16.86)(61.25,11.24)
\psbezier(61.25,11.24)(61.25,11.24)(61.25,11.24)
\psline(61.25,11.24)(61.25,12.5)
}
\rput(63.75,13.74){$\ortt$}
\psline[linewidth=0.35,border=0.45](15.62,-13.75)
(15.62,-26.26)(39.38,-50)
\pscustom[linewidth=0.35]{\psline(21.26,-13.76)(-2.5,-13.75)
\psline(-2.5,-13.75)(5.62,-5.63)
\psbezier(5.62,-5.63)(5.62,-5.63)(5.62,-5.63)
\psline(5.62,-5.63)(13.12,-5.62)
\psline(13.12,-5.62)(21.26,-13.76)
\psbezier(21.26,-13.76)(21.26,-13.76)(21.26,-13.76)
\psline(21.26,-13.76)(21.24,-13.76)
}
\psline[linewidth=0.35](9.38,-5.62)(9.38,1.88)
\rput(9.38,-10){$\oovee$}
\psline[linewidth=0.35,border=0.7](49.38,-13.76)
(49.38,-26.26)(25.62,-50)
\psline[linewidth=0.35](63.12,-13.76)
(63.12,-26.26)(39.38,-50)
\pscustom[linewidth=0.35]{\psline(68.76,-13.76)(45,-13.76)
\psline(45,-13.76)(53.12,-5.62)
\psbezier(53.12,-5.62)(53.12,-5.62)(53.12,-5.62)
\psline(53.12,-5.62)(60.62,-5.62)
\psline(60.62,-5.62)(68.76,-13.76)
\psbezier(68.76,-13.76)(68.76,-13.76)(68.76,-13.76)
\psline(68.76,-13.76)(68.74,-13.76)
}
\psline[linewidth=0.35](56.88,-5.62)(56.88,1.88)
\psline[linewidth=0.35](127.49,10.01)(127.49,36.24)
\pscustom[linewidth=0.35]{\psline(133.13,10.01)(127.49,4.38)
\psline(127.49,4.38)(121.87,10.01)
\psline(121.87,10.01)(133.13,10.01)
\psbezier(133.13,10.01)(133.13,10.01)(133.13,10.01)
}
\rput(127.5,7.5){$\zzero$}
\rput(57.5,25.62){$\oowedge$}
\rput(56.88,-10){$\oowedge$}
\psline[linewidth=0.35](1.88,-13.75)
(1.88,-26.26)(25.62,-50)
\rput(120,-10){$\scal$}
\psline[linewidth=0.35,border=1.05](172.5,21.88)
(172.5,9.37)(180,1.87)
\psline[linewidth=0.35](186.25,21.88)
(186.25,9.37)(178.75,1.87)
\rput{0}(179.38,1.88){\psellipse[linewidth=0.35,fillstyle=solid](0,0)(1.48,-1.49)}
\psline[linewidth=0.2](163.75,18.75)
(163.75,-46.88)
(226.25,-46.88)
(226.25,18.75)
(166.25,18.75)
(163.75,18.75)
(163.75,7.5)(163.75,-10)
\pscustom[linewidth=0.35]{\psline(191.88,21.86)(168.12,21.88)
\psline(168.12,21.88)(176.25,30)
\psbezier(176.25,30)(176.25,30)(176.25,30)
\psline(176.25,30)(183.75,30)
\psline(183.75,30)(191.88,21.86)
\psbezier(191.88,21.86)(191.88,21.86)(191.88,21.86)
\psline(191.88,21.86)(191.87,21.86)
}
\psline[linewidth=0.35](180,30)(180,37.5)
\pscustom[linewidth=0.35,fillcolor=white,fillstyle=solid]{\psbezier(183.75,11.25)(183.75,11.25)(183.75,11.25)(183.75,11.25)
\psline(183.75,11.25)(189.37,11.25)
\psline(189.37,11.25)(189.37,16.87)
\psline(189.37,16.87)(183.75,16.87)
\psline(183.75,16.87)(183.75,11.25)
\psbezier(183.75,11.25)(183.75,11.25)(183.75,11.25)
\psline(183.75,11.25)(183.75,12.5)
}
\rput(186.25,13.75){$\ortt$}
\rput(180,25.62){$\oovee$}
\psline[linewidth=0.35](203.75,21.88)
(203.75,9.38)(211.25,1.88)
\psline[linewidth=0.35](217.5,21.88)
(217.5,9.38)(210,1.88)
\rput{0}(210.62,1.88){\psellipse[linewidth=0.35,fillstyle=solid](0,0)(1.49,-1.48)}
\pscustom[linewidth=0.35]{\psline(223.13,21.86)(199.37,21.88)
\psline(199.37,21.88)(207.5,30)
\psbezier(207.5,30)(207.5,30)(207.5,30)
\psline(207.5,30)(215,30)
\psline(215,30)(223.13,21.86)
\psbezier(223.13,21.86)(223.13,21.86)(223.13,21.86)
\psline(223.13,21.86)(223.12,21.86)
}
\psline[linewidth=0.35](211.25,30)(211.25,37.5)
\pscustom[linewidth=0.35,fillcolor=white,fillstyle=solid]{\psbezier(215,11.24)(215,11.24)(215,11.24)(215,11.24)
\psline(215,11.24)(220.62,11.24)
\psline(220.62,11.24)(220.62,16.86)
\psline(220.62,16.86)(215,16.86)
\psline(215,16.86)(215,11.24)
\psbezier(215,11.24)(215,11.24)(215,11.24)
\psline(215,11.24)(215,12.5)
}
\rput(217.5,13.74){$\ortt$}
\rput(211.25,25.62){$\oowedge$}
\psline[linewidth=0.35](179.37,-10.62)(179.38,1.88)
\rput(179.37,-13.12){$\oone$}
\pscustom[linewidth=0.35]{\psline(185,-10.62)(179.37,-16.25)
\psline(179.37,-16.25)(173.75,-10.62)
\psline(173.75,-10.62)(185,-10.62)
\psbezier(185,-10.62)(185,-10.62)(185,-10.62)
}
\psline[linewidth=0.35](210.62,-10.62)(210.62,1.88)
\pscustom[linewidth=0.35]{\psline(216.25,-10.62)(210.62,-16.25)
\psline(210.62,-16.25)(205,-10.62)
\psline(205,-10.62)(216.25,-10.62)
\psbezier(216.25,-10.62)(216.25,-10.62)(216.25,-10.62)
}
\rput(210.62,-13.12){$\zzero$}
\psline[linewidth=0.35](184.38,-28.12)
(194.38,-18.12)
(204.38,-28.12)
(194.38,-38.12)(184.38,-28.12)
\rput(194.38,-28.12){$\scal$}
\rput(145,0){$\EQLS$}
\pscustom[linewidth=0.2,linestyle=dashed,dash=1 1]{\psline(170,-4.99)(170,-40)
\psline(170,-40)(220,-40)
\psline(220,-40)(220,-5)
\psline(220,-5)(170,-4.99)
\psline(170,-4.99)(170,-8.11)
\psbezier(170,-8.11)(170,-8.11)(170,-8.11)
\psline(170,-8.11)(170,-6.87)
}
\pscustom[linewidth=0.2,linestyle=dashed,dash=1 1]{\psline(-3.75,-2.49)(-3.75,-52.5)
\psline(-3.75,-52.5)(70.62,-52.5)
\psline(70.62,-52.5)(70.62,-2.5)
\psline(70.62,-2.5)(-3.75,-2.49)
\psline(-3.75,-2.49)(-3.75,-5.61)
\psbezier(-3.75,-5.61)(-3.75,-5.61)(-3.75,-5.61)
\psline(-3.75,-5.61)(-3.75,-4.37)
}
\rput{0}(39.38,-50){\psellipse[linewidth=0.35,fillstyle=solid](0,0)(1.49,-1.48)}
\rput{0}(25.62,-50){\psellipse[linewidth=0.35,fillstyle=solid](0,0)(1.49,-1.48)}
\end{pspicture}

\eeq
\vspace{5.5\baselineskip}

\noindent where
\newcommand{\scal}{\scriptstyle\omega}
\newcommand{\oovee}{\scriptscriptstyle \ovee}
\newcommand{\oowedge}{\scriptscriptstyle \owedge}
\newcommand{\EQLS}{=}
\def\JPicScale{.25}
\ifx\JPicScale\undefined\def\JPicScale{1}\fi
\psset{unit=\JPicScale mm}
\psset{linewidth=0.3,dotsep=1,hatchwidth=0.3,hatchsep=1.5,shadowsize=1,dimen=middle}
\psset{dotsize=0.7 2.5,dotscale=1 1,fillcolor=black}
\psset{arrowsize=1 2,arrowlength=1,arrowinset=0.25,tbarsize=0.7 5,bracketlength=0.15,rbracketlength=0.15}
\begin{pspicture}(0,0)(92.5,38.36)
\rput{90}(50,36.88){\psellipse[linewidth=0.35,fillstyle=solid](0,0)(1.48,-1.48)}
\psline[linewidth=0.35](-5.62,2.5)
(4.38,12.5)
(14.38,2.5)
(4.38,-7.5)(-5.62,2.5)
\rput(26.88,1.88){$\EQLS$}
\rput{0}(80.63,36.88){\psellipse[linewidth=0.35,fillstyle=solid](0,0)(1.49,-1.48)}
\psline[linewidth=0.35,border=0.45](56.24,21.24)
(56.24,8.74)(71.88,-6.88)
\pscustom[linewidth=0.35]{\psline(61.88,21.24)(38.12,21.24)
\psline(38.12,21.24)(46.24,29.36)
\psbezier(46.24,29.36)(46.24,29.36)(46.24,29.36)
\psline(46.24,29.36)(53.74,29.39)
\psline(53.74,29.39)(61.88,21.24)
\psbezier(61.88,21.24)(61.88,21.24)(61.88,21.24)
\psline(61.88,21.24)(61.86,21.24)
}
\psline[linewidth=0.35](50,29.39)(50,36.89)
\rput(50,25.01){$\oovee$}
\psline[linewidth=0.35,border=0.7](73.14,21.24)
(73.14,8.74)(57.5,-6.88)
\psline[linewidth=0.35](86.88,21.24)
(86.88,8.74)(71.26,-6.88)
\pscustom[linewidth=0.35]{\psline(92.5,21.24)(68.76,21.24)
\psline(68.76,21.24)(76.88,29.39)
\psbezier(76.88,29.39)(76.88,29.39)(76.88,29.39)
\psline(76.88,29.39)(84.38,29.39)
\psline(84.38,29.39)(92.5,21.24)
\psbezier(92.5,21.24)(92.5,21.24)(92.5,21.24)
\psbezier(92.5,21.24)(92.5,21.24)(92.5,21.24)
}
\psline[linewidth=0.35](80.64,29.39)(80.64,36.89)
\rput(80.64,25.01){$\oowedge$}
\psline[linewidth=0.35](42.5,21.24)
(42.5,8.74)(58.12,-6.88)
\rput(4.38,2.5){$\scal$}
\rput{0}(71.25,-6.25){\psellipse[linewidth=0.35,fillstyle=solid](0,0)(1.49,-1.48)}
\rput{0}(57.5,-6.25){\psellipse[linewidth=0.35,fillstyle=solid](0,0)(1.49,-1.48)}
\end{pspicture}

\renewcommand{\oovee}{\ovee}
\renewcommand{\oowedge}{\owedge}
\hspace{-.75em}. But since the second equation in that diagram also holds, the uniqueness part of the pullback condition implies that the factorizations in the dashed rectangles must be equal, i.e.
\beq\label{eq:U-4}
\newcommand{\oone}{\scriptscriptstyle 1}
\newcommand{\zzero}{\scriptscriptstyle 0}
\def\JPicScale{.5}
\ifx\JPicScale\undefined\def\JPicScale{1}\fi
\psset{unit=\JPicScale mm}
\psset{linewidth=0.3,dotsep=1,hatchwidth=0.3,hatchsep=1.5,shadowsize=1,dimen=middle}
\psset{dotsize=0.7 2.5,dotscale=1 1,fillcolor=black}
\psset{arrowsize=1 2,arrowlength=1,arrowinset=0.25,tbarsize=0.7 5,bracketlength=0.15,rbracketlength=0.15}
\begin{pspicture}(0,0)(133.13,23.76)
\psline[linewidth=0.35](112.49,10.01)(112.5,23.12)
\rput(90,0){$\EQLS$}
\rput(112.49,7.51){$\oone$}
\pscustom[linewidth=0.35]{\psline(118.13,10.01)(112.49,4.38)
\psline(112.49,4.38)(106.87,10.01)
\psline(106.87,10.01)(118.13,10.01)
\psbezier(118.13,10.01)(118.13,10.01)(118.13,10.01)
}
\psline[linewidth=0.35,border=0.45](16.24,8.12)
(16.24,-4.39)(40,-28.12)
\pscustom[linewidth=0.35]{\psline(21.88,8.12)(-1.88,8.12)
\psline(-1.88,8.12)(6.24,16.24)
\psbezier(6.24,16.24)(6.24,16.24)(6.24,16.24)
\psline(6.24,16.24)(13.74,16.26)
\psline(13.74,16.26)(21.88,8.12)
\psbezier(21.88,8.12)(21.88,8.12)(21.88,8.12)
\psline(21.88,8.12)(21.86,8.12)
}
\psline[linewidth=0.35](10,16.26)(10,23.76)
\rput(10,11.88){$\oovee$}
\psline[linewidth=0.35,border=0.7](50,8.12)
(50,-4.39)(26.24,-28.12)
\psline[linewidth=0.35](63.74,8.12)
(63.74,-4.39)(40,-28.12)
\pscustom[linewidth=0.35]{\psline(69.38,8.12)(45.62,8.12)
\psline(45.62,8.12)(53.74,16.26)
\psbezier(53.74,16.26)(53.74,16.26)(53.74,16.26)
\psline(53.74,16.26)(61.24,16.26)
\psline(61.24,16.26)(69.38,8.12)
\psbezier(69.38,8.12)(69.38,8.12)(69.38,8.12)
\psline(69.38,8.12)(69.36,8.12)
}
\psline[linewidth=0.35](57.5,16.26)(57.5,23.76)
\psline[linewidth=0.35](127.49,10.01)(127.5,23.12)
\pscustom[linewidth=0.35]{\psline(133.13,10.01)(127.49,4.38)
\psline(127.49,4.38)(121.87,10.01)
\psline(121.87,10.01)(133.13,10.01)
\psbezier(133.13,10.01)(133.13,10.01)(133.13,10.01)
}
\rput(127.5,7.5){$\zzero$}
\rput(57.5,11.88){$\oowedge$}
\psline[linewidth=0.35](2.5,8.12)
(2.5,-4.39)(26.25,-28.12)
\rput(120,-10){$\scal$}
\psline[linewidth=0.35](110,-10)
(120,0)
(130,-10)
(120,-20)(110,-10)
\rput{0}(40,-28.12){\psellipse[linewidth=0.35,fillstyle=solid](0,0)(1.48,-1.48)}
\rput{0}(26.25,-28.12){\psellipse[linewidth=0.35,fillstyle=solid](0,0)(1.48,-1.48)}
\end{pspicture}

\eeq
\vspace{2\baselineskip}

\noindent  Dualizing both sides yields the antispecialty:
\beq\label{eq:U-5}
\newcommand{\oone}{\scriptscriptstyle 1}
\newcommand{\zzero}{\scriptscriptstyle 0}
\def\JPicScale{.5}
\ifx\JPicScale\undefined\def\JPicScale{1}\fi
\psset{unit=\JPicScale mm}
\psset{linewidth=0.3,dotsep=1,hatchwidth=0.3,hatchsep=1.5,shadowsize=1,dimen=middle}
\psset{dotsize=0.7 2.5,dotscale=1 1,fillcolor=black}
\psset{arrowsize=1 2,arrowlength=1,arrowinset=0.25,tbarsize=0.7 5,bracketlength=0.15,rbracketlength=0.15}
\begin{pspicture}(0,0)(65,25)
\psline[linewidth=0.35](55,18.12)(55,25)
\rput(34.38,0){$\EQLS$}
\rput(54.99,15.64){$\oone$}
\pscustom[linewidth=0.35]{\psline(60.63,18.13)(54.99,12.5)
\psline(54.99,12.5)(49.37,18.13)
\psline(49.37,18.13)(60.63,18.13)
\psbezier(60.63,18.13)(60.63,18.13)(60.63,18.13)
}
\psline[linewidth=0.35,border=0.45](16.24,8.12)
(16.24,-4.39)(16.25,-4.38)
\pscustom[linewidth=0.35]{\psline(21.88,8.12)(-1.88,8.12)
\psline(-1.88,8.12)(6.24,16.24)
\psbezier(6.24,16.24)(6.24,16.24)(6.24,16.24)
\psline(6.24,16.24)(13.74,16.26)
\psline(13.74,16.26)(21.88,8.12)
\psbezier(21.88,8.12)(21.88,8.12)(21.88,8.12)
\psline(21.88,8.12)(21.86,8.12)
}
\psline[linewidth=0.35](10,16.26)(10,25)
\rput(10,11.88){$\oovee$}
\pscustom[linewidth=0.35]{\psline(-1.89,-5)(21.88,-5)
\psline(21.88,-5)(13.76,-13.14)
\psbezier(13.76,-13.14)(13.76,-13.14)(13.76,-13.14)
\psline(13.76,-13.14)(6.26,-13.14)
\psline(6.26,-13.14)(-1.89,-5)
\psbezier(-1.89,-5)(-1.89,-5)(-1.89,-5)
\psline(-1.89,-5)(-1.86,-5)
}
\psline[linewidth=0.35](10,-24.38)(10,-13.12)
\psline[linewidth=0.35](55,-18.12)(55,-25)
\pscustom[linewidth=0.35]{\psline(49.37,-18.13)(55.01,-12.51)
\psline(55.01,-12.51)(60.63,-18.13)
\psline(60.63,-18.13)(49.37,-18.13)
\psbezier(49.37,-18.13)(49.37,-18.13)(49.37,-18.13)
}
\rput(55,-15.62){$\zzero$}
\rput(10,-9.38){$\oowedge$}
\psline[linewidth=0.35](2.5,8.12)
(2.5,-4.39)(2.5,-4.38)
\psline[linewidth=0.35](45,0)
(55,10)
(65,0)
(55,-10)(45,0)
\rput(55,0){$\scal$}
\end{pspicture}

\eeq
\vspace{1\baselineskip}

\noindent 
To complete the proof, we proceed to transform the left-hand side of \eqref{eq:U-1}. 
Since $\ovee$ and $\owedge$ are single-valued, Prop.\ref{prop:maps-char}ii. says that we can we can distribute each of them above the black dots on the left-hand side of \eqref{eq:U-1}. Applying the associativity, the left-hand side of \eqref{eq:U-1} is transformed into the left-hand side of the following equation.

\beq\label{eq:U-6}
\newcommand{\oone}{\scriptscriptstyle 1}
\newcommand{\zzero}{\scriptscriptstyle 0}
\def\JPicScale{.5}
\ifx\JPicScale\undefined\def\JPicScale{1}\fi
\psset{unit=\JPicScale mm}
\psset{linewidth=0.3,dotsep=1,hatchwidth=0.3,hatchsep=1.5,shadowsize=1,dimen=middle}
\psset{dotsize=0.7 2.5,dotscale=1 1,fillcolor=black}
\psset{arrowsize=1 2,arrowlength=1,arrowinset=0.25,tbarsize=0.7 5,bracketlength=0.15,rbracketlength=0.15}
\begin{pspicture}(0,0)(187.5,55.62)
\psline[linewidth=0.35](25,21.88)
(25,9.37)(25,9.38)
\pscustom[linewidth=0.2]{\psline(-0.62,35)(-0.62,-50.62)
\psline(-0.62,-50.62)(86.25,-50.62)
\psline(86.25,-50.62)(86.25,35)
\psline(86.25,35)(-0.62,35)
\psline(-0.62,35)(-0.62,31.88)
\psbezier(-0.62,31.88)(-0.62,31.88)(-0.62,31.88)
\psline(-0.62,31.88)(-0.62,33.12)
}
\pscustom[linewidth=0.35]{\psline(30.63,21.86)(6.87,21.88)
\psline(6.87,21.88)(15,30)
\psbezier(15,30)(15,30)(15,30)
\psline(15,30)(22.5,30)
\psline(22.5,30)(30.63,21.86)
\psbezier(30.63,21.86)(30.63,21.86)(30.63,21.86)
\psline(30.63,21.86)(30.62,21.86)
}
\psline[linewidth=0.35](18.75,30)(18.75,37.5)
\pscustom[linewidth=0.35,fillcolor=white,fillstyle=solid]{\psbezier(22.5,11.25)(22.5,11.25)(22.5,11.25)(22.5,11.25)
\psline(22.5,11.25)(28.12,11.25)
\psline(28.12,11.25)(28.12,16.87)
\psline(28.12,16.87)(22.5,16.87)
\psline(22.5,16.87)(22.5,11.25)
\psbezier(22.5,11.25)(22.5,11.25)(22.5,11.25)
\psline(22.5,11.25)(22.5,12.5)
}
\rput(25,13.75){$\ortt$}
\rput(98.75,0){$\EQLS$}
\rput(18.75,25.62){$\oovee$}
\psline[linewidth=0.35](58.75,21.88)
(58.75,9.38)(58.75,-40.62)
\pscustom[linewidth=0.35]{\psline(72.5,21.88)(72.5,9.38)
\psbezier(72.5,9.38)(72.5,9.38)(72.5,9.38)
}
\pscustom[linewidth=0.35]{\psline(78.13,21.86)(54.37,21.88)
\psline(54.37,21.88)(62.5,30)
\psbezier(62.5,30)(62.5,30)(62.5,30)
\psline(62.5,30)(70,30)
\psline(70,30)(78.13,21.86)
\psbezier(78.13,21.86)(78.13,21.86)(78.13,21.86)
\psline(78.13,21.86)(78.12,21.86)
}
\psline[linewidth=0.35](66.25,30)(66.25,37.5)
\pscustom[linewidth=0.35,fillcolor=white,fillstyle=solid]{\psbezier(70,11.24)(70,11.24)(70,11.24)(70,11.24)
\psline(70,11.24)(75.62,11.24)
\psline(75.62,11.24)(75.62,16.86)
\psline(75.62,16.86)(70,16.86)
\psline(70,16.86)(70,11.24)
\psbezier(70,11.24)(70,11.24)(70,11.24)
\psline(70,11.24)(70,12.5)
}
\rput(72.5,13.74){$\ortt$}
\psline[linewidth=0.35,border=0.45](11.88,21.88)
(11.88,-6.25)(52.5,-46.88)
\pscustom[linewidth=0.35]{\psline(36.89,-1.88)(13.12,-1.88)
\psline(13.12,-1.88)(21.24,6.24)
\psbezier(21.24,6.24)(21.24,6.24)(21.24,6.24)
\psline(21.24,6.24)(28.74,6.26)
\psline(28.74,6.26)(36.89,-1.88)
\psbezier(36.89,-1.88)(36.89,-1.88)(36.89,-1.88)
\psline(36.89,-1.88)(36.86,-1.88)
}
\psline[linewidth=0.35](25.01,6.26)(25,11.25)
\rput(25.01,1.88){$\oovee$}
\psline[linewidth=0.35](80,-2.5)
(80,-19.38)(52.51,-46.86)
\pscustom[linewidth=0.35]{\psline(84.38,-1.88)(60.62,-1.88)
\psline(60.62,-1.88)(68.74,6.26)
\psbezier(68.74,6.26)(68.74,6.26)(68.74,6.26)
\psline(68.74,6.26)(76.24,6.26)
\psline(76.24,6.26)(84.38,-1.88)
\psbezier(84.38,-1.88)(84.38,-1.88)(84.38,-1.88)
\psline(84.38,-1.88)(84.36,-1.88)
}
\psline[linewidth=0.35](72.5,6.26)(72.5,10)
\rput(66.25,25.62){$\oowedge$}
\rput(72.5,1.88){$\oowedge$}
\pscustom[linewidth=0.35]{\psline(18.75,40)(18.75,30)
\psbezier(18.75,30)(18.75,30)(18.75,30)
}
\pscustom[linewidth=0.35]{\psline(24.38,39.98)(0.62,40)
\psline(0.62,40)(8.75,48.12)
\psbezier(8.75,48.12)(8.75,48.12)(8.75,48.12)
\psline(8.75,48.12)(16.25,48.12)
\psline(16.25,48.12)(24.38,39.98)
\psbezier(24.38,39.98)(24.38,39.98)(24.38,39.98)
\psline(24.38,39.98)(24.37,39.98)
}
\psline[linewidth=0.35](12.5,48.12)(12.5,55.62)
\rput(12.5,43.75){$\oovee$}
\pscustom[linewidth=0.35]{\psline(66.25,40)(66.25,30)
\psbezier(66.25,30)(66.25,30)(66.25,30)
}
\pscustom[linewidth=0.35]{\psline(71.88,39.98)(48.12,40)
\psline(48.12,40)(56.25,48.12)
\psbezier(56.25,48.12)(56.25,48.12)(56.25,48.12)
\psline(56.25,48.12)(63.75,48.12)
\psline(63.75,48.12)(71.88,39.98)
\psbezier(71.88,39.98)(71.88,39.98)(71.88,39.98)
\psline(71.88,39.98)(71.87,39.98)
}
\psline[linewidth=0.35](60,48.12)(60,55.62)
\rput(60,43.75){$\oowedge$}
\psline[linewidth=0.35](31.25,-1.88)
(31.25,-14.38)(31.25,-26.25)
\psline[linewidth=0.35](4.38,39.38)
(4.38,-23.76)(28.12,-47.5)
\psline[linewidth=0.35,border=0.7](64.38,-2.5)
(64.38,-11.88)(28.75,-47.5)
\psline[linewidth=0.35](51.88,39.38)
(51.88,26.88)(51.88,-23.12)
\psline[linewidth=0.35,border=1.05](18.75,-2.5)
(18.75,-15.62)(18.75,-36.88)
\pscustom[linewidth=0.35]{\psline(135.63,39.98)(111.88,40)
\psline(111.88,40)(120,48.12)
\psbezier(120,48.12)(120,48.12)(120,48.12)
\psline(120,48.12)(127.5,48.12)
\psline(127.5,48.12)(135.63,39.98)
\psbezier(135.63,39.98)(135.63,39.98)(135.63,39.98)
\psline(135.63,39.98)(135.62,39.98)
}
\psline[linewidth=0.35](123.75,48.12)(123.75,55.62)
\rput(123.75,43.75){$\oovee$}
\pscustom[linewidth=0.35]{\psline(183.13,39.98)(159.37,40)
\psline(159.37,40)(167.5,48.12)
\psbezier(167.5,48.12)(167.5,48.12)(167.5,48.12)
\psline(167.5,48.12)(175,48.12)
\psline(175,48.12)(183.13,39.98)
\psbezier(183.13,39.98)(183.13,39.98)(183.13,39.98)
\psline(183.13,39.98)(183.12,39.98)
}
\psline[linewidth=0.35](171.25,48.12)(171.25,55.62)
\rput(171.25,43.75){$\oowedge$}
\psline[linewidth=0.35](115.63,39.38)
(115.62,-10)(146.25,-40.62)
\psline[linewidth=0.35](177.5,40)
(177.5,-10)(146.25,-41.25)
\psline[linewidth=0.35](164.38,39.38)
(164.38,26.88)(164.38,-23.12)
\psline[linewidth=0.35](130,40)
(130,-15.62)(130,-24.38)
\pscustom[linewidth=0.35,fillcolor=white,fillstyle=solid]{\psbezier(127.5,11.88)(127.5,11.88)(127.5,11.88)(127.5,11.88)
\psline(127.5,11.88)(133.12,11.88)
\psline(133.12,11.88)(133.12,17.5)
\psline(133.12,17.5)(127.5,17.5)
\psline(127.5,17.5)(127.5,11.88)
\psbezier(127.5,11.88)(127.5,11.88)(127.5,11.88)
\psline(127.5,11.88)(127.5,13.12)
}
\rput(130,14.38){$\ortt$}
\pscustom[linewidth=0.35,fillcolor=white,fillstyle=solid]{\psbezier(175,11.87)(175,11.87)(175,11.87)(175,11.87)
\psline(175,11.87)(180.62,11.87)
\psline(180.62,11.87)(180.62,17.48)
\psline(180.62,17.48)(175,17.48)
\psline(175,17.48)(175,11.87)
\psbezier(175,11.87)(175,11.87)(175,11.87)
\psline(175,11.87)(175,13.12)
}
\rput(177.5,14.37){$\ortt$}
\pscustom[linewidth=0.2]{\psline(110,34.38)(110,-50.62)
\psline(110,-50.62)(187.5,-50.62)
\psline(187.5,-50.62)(187.5,35)
\psline(187.5,35)(110,35)
\psline(110,35)(110,31.88)
\psbezier(110,31.88)(110,31.88)(110,31.88)
\psline(110,31.88)(110,33.12)
}
\rput{0}(31.25,-25.39){\psellipse[linewidth=0.35,fillstyle=solid](0,0)(1.49,-1.49)}
\rput{0}(51.88,-24.38){\psellipse[linewidth=0.35,fillstyle=solid](0,0)(1.48,-1.48)}
\rput{0}(18.75,-38.12){\psellipse[linewidth=0.35,fillstyle=solid](0,0)(1.49,-1.49)}
\rput{0}(58.75,-40){\psellipse[linewidth=0.35,fillstyle=solid](0,0)(1.49,-1.48)}
\rput{90}(130,-24.38){\psellipse[linewidth=0.35,fillstyle=solid](0,0)(1.49,-1.48)}
\rput{90}(164.38,-23.75){\psellipse[linewidth=0.35,fillstyle=solid](0,0)(1.49,-1.49)}
\rput{0}(28.75,-47.5){\psellipse[linewidth=0.35,fillstyle=solid](0,0)(1.49,-1.49)}
\rput{0}(52.5,-47.5){\psellipse[linewidth=0.35,fillstyle=solid](0,0)(1.49,-1.49)}
\rput{0}(146.88,-41.25){\psellipse[linewidth=0.35,fillstyle=solid](0,0)(1.49,-1.49)}
\end{pspicture}

\eeq
\vspace{4\baselineskip}

\noindent The right-hand side is a path around the third pullback in \eqref{eq:pb-1}. Factoring the left-hand side through the pullback, postcomposing one of the branches with $\ortt$, and reducing $\owedge$ to $\ovee$ precomposed and postcomposed with $\ortt$s, we get
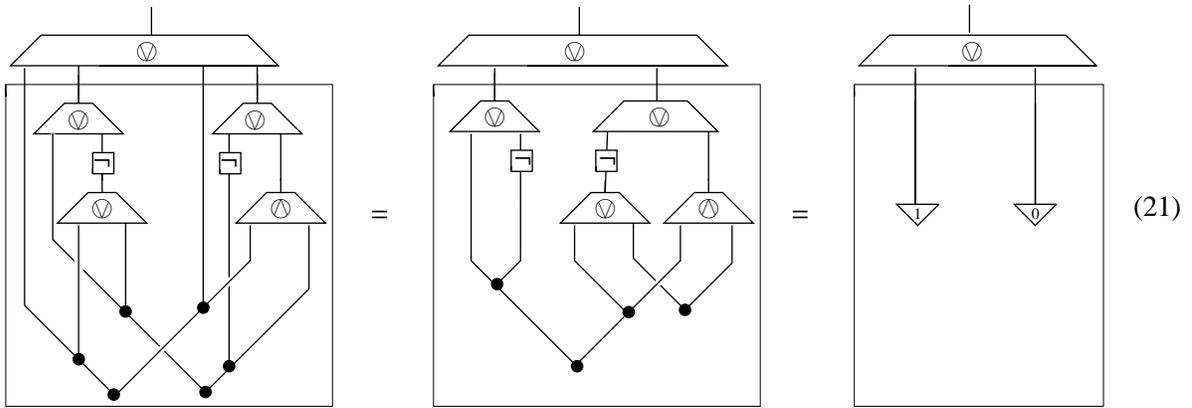
\beq\label{eq:U-7}
\newcommand{\oone}{\scriptscriptstyle 1}
\newcommand{\zzero}{\scriptscriptstyle 0}
\def\JPicScale{.5}
\ifx\JPicScale\undefined\def\JPicScale{1}\fi
\psset{unit=\JPicScale mm}
\psset{linewidth=0.3,dotsep=1,hatchwidth=0.3,hatchsep=1.5,shadowsize=1,dimen=middle}
\psset{dotsize=0.7 2.5,dotscale=1 1,fillcolor=black}
\psset{arrowsize=1 2,arrowlength=1,arrowinset=0.25,tbarsize=0.7 5,bracketlength=0.15,rbracketlength=0.15}
\begin{pspicture}(0,0)(291.25,56.25)
\psline[linewidth=0.35](25,21.88)
(25,9.37)(25,9.38)
\pscustom[linewidth=0.2]{\psline(-0.62,35)(-0.62,-50.62)
\psline(-0.62,-50.62)(86.25,-50.62)
\psline(86.25,-50.62)(86.25,35)
\psline(86.25,35)(-0.62,35)
\psline(-0.62,35)(-0.62,31.88)
\psbezier(-0.62,31.88)(-0.62,31.88)(-0.62,31.88)
\psline(-0.62,31.88)(-0.62,33.12)
}
\pscustom[linewidth=0.35]{\psline(30.63,21.86)(6.87,21.88)
\psline(6.87,21.88)(15,30)
\psbezier(15,30)(15,30)(15,30)
\psline(15,30)(22.5,30)
\psline(22.5,30)(30.63,21.86)
\psbezier(30.63,21.86)(30.63,21.86)(30.63,21.86)
\psline(30.63,21.86)(30.62,21.86)
}
\psline[linewidth=0.35](18.75,30)(18.75,37.5)
\pscustom[linewidth=0.35,fillcolor=white,fillstyle=solid]{\psbezier(22.5,11.25)(22.5,11.25)(22.5,11.25)(22.5,11.25)
\psline(22.5,11.25)(28.12,11.25)
\psline(28.12,11.25)(28.12,16.87)
\psline(28.12,16.87)(22.5,16.87)
\psline(22.5,16.87)(22.5,11.25)
\psbezier(22.5,11.25)(22.5,11.25)(22.5,11.25)
\psline(22.5,11.25)(22.5,12.5)
}
\rput(25,13.75){$\ortt$}
\rput(98.75,0){$\EQLS$}
\rput(18.75,25.62){$\oovee$}
\psline[linewidth=0.35](58.75,21.88)
(58.75,9.38)(58.75,-40.62)
\pscustom[linewidth=0.35]{\psline(72.5,21.88)(72.5,9.38)
\psbezier(72.5,9.38)(72.5,9.38)(72.5,9.38)
}
\pscustom[linewidth=0.35]{\psline(78.13,21.86)(54.37,21.88)
\psline(54.37,21.88)(62.5,30)
\psbezier(62.5,30)(62.5,30)(62.5,30)
\psline(62.5,30)(70,30)
\psline(70,30)(78.13,21.86)
\psbezier(78.13,21.86)(78.13,21.86)(78.13,21.86)
\psline(78.13,21.86)(78.12,21.86)
}
\psline[linewidth=0.35](66.25,30)(66.25,37.5)
\pscustom[linewidth=0.35,fillcolor=white,fillstyle=solid]{\psbezier(56.25,11.24)(56.25,11.24)(56.25,11.24)(56.25,11.24)
\psline(56.25,11.24)(61.87,11.24)
\psline(61.87,11.24)(61.87,16.86)
\psline(61.87,16.86)(56.25,16.86)
\psline(56.25,16.86)(56.25,11.24)
\psbezier(56.25,11.24)(56.25,11.24)(56.25,11.24)
\psline(56.25,11.24)(56.25,12.5)
}
\rput(58.75,13.74){$\ortt$}
\psline[linewidth=0.35,border=0.45](11.88,21.88)
(11.88,-6.25)(52.5,-46.88)
\pscustom[linewidth=0.35]{\psline(36.89,-1.88)(13.12,-1.88)
\psline(13.12,-1.88)(21.24,6.24)
\psbezier(21.24,6.24)(21.24,6.24)(21.24,6.24)
\psline(21.24,6.24)(28.74,6.26)
\psline(28.74,6.26)(36.89,-1.88)
\psbezier(36.89,-1.88)(36.89,-1.88)(36.89,-1.88)
\psline(36.89,-1.88)(36.86,-1.88)
}
\psline[linewidth=0.35](25.01,6.26)(25,11.25)
\rput(25.01,1.88){$\oovee$}
\psline[linewidth=0.35](80,-2.5)
(80,-19.38)(52.51,-46.86)
\pscustom[linewidth=0.35]{\psline(84.38,-1.88)(60.62,-1.88)
\psline(60.62,-1.88)(68.74,6.26)
\psbezier(68.74,6.26)(68.74,6.26)(68.74,6.26)
\psline(68.74,6.26)(76.24,6.26)
\psline(76.24,6.26)(84.38,-1.88)
\psbezier(84.38,-1.88)(84.38,-1.88)(84.38,-1.88)
\psline(84.38,-1.88)(84.36,-1.88)
}
\psline[linewidth=0.35](72.5,6.26)(72.5,10)
\rput(72.5,1.88){$\oowedge$}
\pscustom[linewidth=0.35]{\psline(18.75,40)(18.75,30)
\psbezier(18.75,30)(18.75,30)(18.75,30)
}
\pscustom[linewidth=0.35]{\psline(71.25,40)(0.62,40)
\psline(0.62,40)(8.75,48.12)
\psbezier(8.75,48.12)(8.75,48.12)(8.75,48.12)
\psline(8.75,48.12)(63.75,48.12)
\psline(63.75,48.12)(71.88,40)
\psline(71.88,40)(24.38,39.98)
\psline(24.38,39.98)(24.37,39.98)
}
\psline[linewidth=0.35](38.12,48.12)(38.12,55.62)
\rput(36.88,43.75){$\oovee$}
\pscustom[linewidth=0.35]{\psline(66.25,40)(66.25,30)
\psbezier(66.25,30)(66.25,30)(66.25,30)
}
\psline[linewidth=0.35](31.25,-1.88)
(31.25,-14.38)(31.25,-26.25)
\psline[linewidth=0.35](4.38,39.38)
(4.38,-23.76)(28.12,-47.5)
\psline[linewidth=0.35,border=0.7](64.38,-2.5)
(64.38,-11.88)(28.75,-47.5)
\psline[linewidth=0.35](51.88,40)
(51.88,26.88)(51.88,-23.12)
\psline[linewidth=0.35,border=1.05](18.75,-2.5)
(18.75,-15.62)(18.75,-36.88)
\rput(65.62,25.62){$\oovee$}
\rput{0}(31.25,-25.39){\psellipse[linewidth=0.35,fillstyle=solid](0,0)(1.49,-1.49)}
\rput{0}(51.88,-24.38){\psellipse[linewidth=0.35,fillstyle=solid](0,0)(1.48,-1.48)}
\rput{0}(18.75,-38.12){\psellipse[linewidth=0.35,fillstyle=solid](0,0)(1.49,-1.49)}
\rput{0}(58.75,-40){\psellipse[linewidth=0.35,fillstyle=solid](0,0)(1.49,-1.48)}
\psline[linewidth=0.35](159.38,22.5)
(158.75,6.25)(158.75,9.38)
\pscustom[linewidth=0.2]{\psline(113.12,35)(113.12,-50.62)
\psline(113.12,-50.62)(200,-50.62)
\psline(200,-50.62)(200,35)
\psline(200,35)(113.12,35)
\psline(113.12,35)(113.12,31.88)
\psbezier(113.12,31.88)(113.12,31.88)(113.12,31.88)
\psline(113.12,31.88)(113.12,33.12)
}
\pscustom[linewidth=0.35,fillcolor=white,fillstyle=solid]{\psbezier(156.26,11.88)(156.26,11.88)(156.26,11.88)(156.26,11.88)
\psline(156.26,11.88)(161.88,11.88)
\psline(161.88,11.88)(161.88,17.5)
\psline(161.88,17.5)(156.26,17.5)
\psline(156.26,17.5)(156.26,11.88)
\psbezier(156.26,11.88)(156.26,11.88)(156.26,11.88)
\psline(156.26,11.88)(156.26,13.12)
}
\rput(159.38,14.38){$\ortt$}
\psline[linewidth=0.35](136.26,21.88)
(136.25,-11.88)(129.38,-18.75)
\pscustom[linewidth=0.35]{\psline(186.25,22.5)(186.25,9.38)
\psbezier(186.25,9.38)(186.25,9.38)(186.25,9.38)
}
\pscustom[linewidth=0.35,fillcolor=white,fillstyle=solid]{\psbezier(133.76,11.88)(133.76,11.88)(133.76,11.88)(133.76,11.88)
\psline(133.76,11.88)(139.38,11.88)
\psline(139.38,11.88)(139.38,17.5)
\psline(139.38,17.5)(133.76,17.5)
\psline(133.76,17.5)(133.76,11.88)
\psbezier(133.76,11.88)(133.76,11.88)(133.76,11.88)
\psline(133.76,11.88)(133.76,13.14)
}
\rput(136.87,14.37){$\ortt$}
\psline[linewidth=0.35,border=0.45](166.25,-1.88)
(166.25,-11.25)(180.62,-25.62)
\pscustom[linewidth=0.35]{\psline(170.64,-1.88)(146.88,-1.88)
\psline(146.88,-1.88)(155,6.25)
\psbezier(155,6.25)(155,6.25)(155,6.25)
\psline(155,6.25)(162.5,6.25)
\psline(162.5,6.25)(170.64,-1.88)
\psbezier(170.64,-1.88)(170.64,-1.88)(170.64,-1.88)
\psline(170.64,-1.88)(170.62,-1.88)
}
\rput(158.76,1.88){$\oovee$}
\psline[linewidth=0.35](192.5,-2.5)
(192.5,-12.5)(179.38,-25.62)
\pscustom[linewidth=0.35]{\psline(198.14,-1.88)(174.38,-1.88)
\psline(174.38,-1.88)(182.5,6.25)
\psbezier(182.5,6.25)(182.5,6.25)(182.5,6.25)
\psline(182.5,6.25)(190,6.25)
\psline(190,6.25)(198.14,-1.88)
\psbezier(198.14,-1.88)(198.14,-1.88)(198.14,-1.88)
\psline(198.14,-1.88)(198.12,-1.88)
}
\psline[linewidth=0.35](186.26,6.25)(186.25,10)
\rput(186.26,1.88){$\oowedge$}
\pscustom[linewidth=0.35]{\psline(190.62,40)(114.38,40)
\psline(114.38,40)(122.5,48.12)
\psbezier(122.5,48.12)(122.5,48.12)(122.5,48.12)
\psline(122.5,48.12)(183.12,48.12)
\psline(183.12,48.12)(191.25,40)
\psline(191.25,40)(138.13,39.98)
\psline(138.13,39.98)(138.12,39.98)
}
\psline[linewidth=0.35](151.88,48.12)(151.88,55.62)
\rput(150.62,43.75){$\oovee$}
\psline[linewidth=0.35,border=0.7](178.76,-2.5)
(178.76,-11.88)(150.62,-40)
\psline[linewidth=0.35,border=1.05](150.62,-2.5)
(150.62,-11.88)(165.62,-26.88)
\rput{0}(165,-25.62){\psellipse[linewidth=0.35,fillstyle=solid](0,0)(1.49,-1.48)}
\rput{90}(130,-18.13){\psellipse[linewidth=0.35,fillstyle=solid](0,0)(1.49,-1.48)}
\psline[linewidth=0.35](123.12,21.88)
(123.12,-11.88)(151.25,-40)
\pscustom[linewidth=0.35]{\psline(141.26,22.48)(117.5,22.5)
\psline(117.5,22.5)(125.62,30.62)
\psbezier(125.62,30.62)(125.62,30.62)(125.62,30.62)
\psline(125.62,30.62)(133.12,30.62)
\psline(133.12,30.62)(141.26,22.48)
\psbezier(141.26,22.48)(141.26,22.48)(141.26,22.48)
\psline(141.26,22.48)(141.24,22.48)
}
\psline[linewidth=0.35](129.38,30.62)(129.38,38.12)
\rput(129.38,26.25){$\oovee$}
\pscustom[linewidth=0.35]{\psline(188.76,22.48)(155.62,22.5)
\psline(155.62,22.5)(163.75,30.62)
\psline(163.75,30.62)(173.12,30.62)
\psline(173.12,30.62)(180.62,30.62)
\psline(180.62,30.62)(188.76,22.48)
\psbezier(188.76,22.48)(188.76,22.48)(188.76,22.48)
\psline(188.76,22.48)(188.74,22.48)
}
\pscustom[linewidth=0.35]{\psline(129.38,39.38)(129.38,30.62)
\psbezier(129.38,30.62)(129.38,30.62)(129.38,30.62)
}
\pscustom[linewidth=0.35]{\psline(172.5,39.38)(172.5,30.62)
\psbezier(172.5,30.62)(172.5,30.62)(172.5,30.62)
}
\rput(173.12,26.25){$\oovee$}
\rput(210.62,0){$\EQLS$}
\pscustom[linewidth=0.2]{\psline(225,35)(225,-50.62)
\psline(225,-50.62)(291.25,-50.62)
\psline(291.25,-50.62)(291.25,35)
\psline(291.25,35)(225,35)
\psline(225,35)(225,31.88)
\psbezier(225,31.88)(225,31.88)(225,31.88)
\psline(225,31.88)(225,33.12)
}
\pscustom[linewidth=0.35]{\psline(289.38,40)(226.26,40)
\psline(226.26,40)(234.38,48.12)
\psbezier(234.38,48.12)(234.38,48.12)(234.38,48.12)
\psline(234.38,48.12)(281.25,48.12)
\psline(281.25,48.12)(289.38,40)
\psline(289.38,40)(250,39.98)
\psbezier(250,39.98)(250,39.98)(250,39.98)
}
\psline[linewidth=0.35](255.62,48.75)(255.62,56.25)
\rput(256.25,43.75){$\oovee$}
\psline[linewidth=0.35](241.25,30.62)(241.25,38.12)
\psline[linewidth=0.35](241.25,39.38)
(241.25,3.12)(241.25,30.62)
\psline[linewidth=0.35](273.12,39.38)
(273.12,3.12)(273.12,30)
\rput(241.86,0.63){$\oone$}
\pscustom[linewidth=0.35]{\psline(247.5,3.13)(241.86,-2.5)
\psline(241.86,-2.5)(236.24,3.13)
\psline(236.24,3.13)(247.5,3.13)
\psbezier(247.5,3.13)(247.5,3.13)(247.5,3.13)
}
\pscustom[linewidth=0.35]{\psline(278.76,3.13)(273.12,-2.5)
\psline(273.12,-2.5)(267.5,3.13)
\psline(267.5,3.13)(278.76,3.13)
\psbezier(278.76,3.13)(278.76,3.13)(278.76,3.13)
}
\rput(273.12,0.62){$\zzero$}
\rput{90}(28.12,-47.5){\psellipse[linewidth=0.35,fillstyle=solid](0,0)(1.49,-1.48)}
\rput{0}(52.5,-46.88){\psellipse[linewidth=0.35,fillstyle=solid](0,0)(1.49,-1.48)}
\rput{0}(151.25,-40){\psellipse[linewidth=0.35,fillstyle=solid](0,0)(1.49,-1.49)}
\rput{0}(180,-25){\psellipse[linewidth=0.35,fillstyle=solid](0,0)(1.49,-1.49)}
\end{pspicture}

\eeq
\vspace{4\baselineskip}

\noindent from which the result follows using the second pullback of \eqref{eq:pb} and \eqref{eq:pb-string}.
\end{proof}
%
%
%

\section{Frobenius and modularity}\label{Sec:Frob}

In lattice theory, the modularity condition is usually written in the form
\bear
x\leq z &\Longrightarrow & (x\vee y) \wedge z =
x \vee (y\wedge z)  
\eear
In an effect algebra, $x\ovee y$ is defined if and only if $x\leq \ort y$, whereas $y\owedge z$ is defined if and only if $\ort y \leq z$, where $u\leq w$ abbreviates $\exists v.\  u\ovee v = w$. Both $x\ovee y$ and $y\owedge z$ are thus  defined if and only if $x\leq \ort y \leq z$. The modularity law for effect algebras is thus
\bea\label{eq:modular}
x\leq \ort y \leq z &\Longrightarrow & (x\ovee y) \owedge z = x \ovee (y\owedge z)  
\eea
The following definition, stated in an arbitrary dagger-compact category $\CCc$, is equivalent to \eqref{eq:modular} when restricted to partial functions, i.e. to single-valued morphisms in $\CCc=\Rel$.

\begin{definition}
A convolution algebra $(A,\ovee,0,\owedge,1)$ over a self-dual object $A$ in a dagger-compact category $\CCc$ is said to be \emph{modular} when the following equation holds
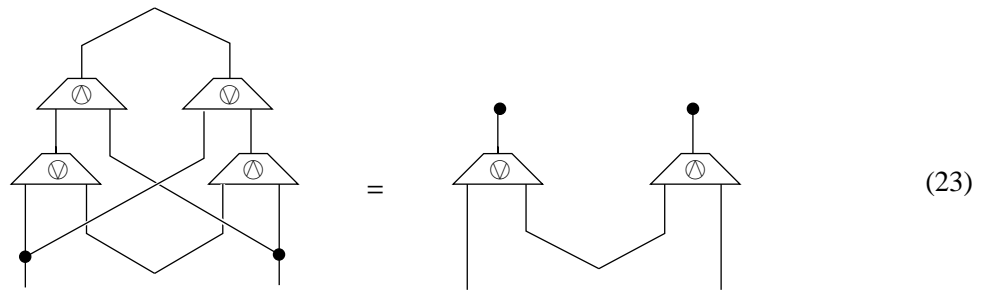
\beq\label{eq:Mod}
\newcommand{\oovee}{\ovee}
\newcommand{\oowedge}{\owedge}
\newcommand{\EQLS}{=}
\def\JPicScale{.5}
\ifx\JPicScale\undefined\def\JPicScale{1}\fi
\psset{unit=\JPicScale mm}
\psset{linewidth=0.3,dotsep=1,hatchwidth=0.3,hatchsep=1.5,shadowsize=1,dimen=middle}
\psset{dotsize=0.7 2.5,dotscale=1 1,fillcolor=black}
\psset{arrowsize=1 2,arrowlength=1,arrowinset=0.25,tbarsize=0.7 5,bracketlength=0.15,rbracketlength=0.15}
\begin{pspicture}(0,0)(195.63,48.75)
\psline[linewidth=0.35](13.75,21.88)
(13.75,10)(13.75,11.88)
\rput(98.75,0){$\EQLS$}
\rput(60,25.62){$\oovee$}
\pscustom[linewidth=0.35]{\psline(65.62,21.88)(65.62,10)
\psbezier(65.62,10)(65.62,10)(65.62,10)
}
\pscustom[linewidth=0.35]{\psline(25.64,1.87)(1.87,1.87)
\psline(1.87,1.87)(9.99,9.99)
\psbezier(9.99,9.99)(9.99,9.99)(9.99,9.99)
\psline(9.99,9.99)(17.49,10.01)
\psline(17.49,10.01)(25.64,1.87)
\psbezier(25.64,1.87)(25.64,1.87)(25.64,1.87)
\psline(25.64,1.87)(25.61,1.87)
}
\rput(14.38,5.62){$\oovee$}
\pscustom[linewidth=0.35]{\psline(78.13,1.87)(54.37,1.87)
\psline(54.37,1.87)(62.49,10.01)
\psbezier(62.49,10.01)(62.49,10.01)(62.49,10.01)
\psline(62.49,10.01)(69.99,10.01)
\psline(69.99,10.01)(78.13,1.87)
\psbezier(78.13,1.87)(78.13,1.87)(78.13,1.87)
\psline(78.13,1.87)(78.11,1.87)
}
\rput(20.62,25.62){$\oowedge$}
\psline[linewidth=0.35](5.63,1.88)
(5.62,-20.62)(5.62,-25.62)
\psline[linewidth=0.35](73.12,1.88)
(73.12,-11.24)(73.12,-25)
\rput{0}(183.13,21.88){\psellipse[linewidth=0.35,fillstyle=solid](0,0)(1.49,-1.48)}
\rput{0}(131.88,21.88){\psellipse[linewidth=0.35,fillstyle=solid](0,0)(1.48,-1.49)}
\rput(66.25,5.62){$\oowedge$}
\psline[linewidth=0.35,border=0.45](58.12,1.88)
(58.12,-11.25)(40,-21.88)
\psline[linewidth=0.35,border=0.5](28.12,21.88)
(28.12,9.38)(73.12,-16.88)
\rput{0}(73.12,-16.88){\psellipse[linewidth=0.35,fillstyle=solid](0,0)(1.48,-1.48)}
\pscustom[linewidth=0.35]{\psline(32.51,21.86)(8.75,21.88)
\psline(8.75,21.88)(16.88,30)
\psbezier(16.88,30)(16.88,30)(16.88,30)
\psline(16.88,30)(24.38,30)
\psline(24.38,30)(32.51,21.86)
\psbezier(32.51,21.86)(32.51,21.86)(32.51,21.86)
\psline(32.51,21.86)(32.5,21.86)
}
\pscustom[linewidth=0.35]{\psline(71.26,21.86)(47.5,21.88)
\psline(47.5,21.88)(55.63,30)
\psbezier(55.63,30)(55.63,30)(55.63,30)
\psline(55.63,30)(63.13,30)
\psline(63.13,30)(71.26,21.86)
\psbezier(71.26,21.86)(71.26,21.86)(71.26,21.86)
\psline(71.26,21.86)(71.25,21.86)
}
\psline[linewidth=0.35](131.25,21.88)
(131.25,10)(131.25,11.88)
\pscustom[linewidth=0.35]{\psline(183.12,21.88)(183.12,10)
\psbezier(183.12,10)(183.12,10)(183.12,10)
}
\psline[linewidth=0.35](138.75,1.88)
(138.75,-10.62)(158.12,-20.62)
\pscustom[linewidth=0.35]{\psline(143.14,1.87)(119.37,1.87)
\psline(119.37,1.87)(127.49,9.99)
\psbezier(127.49,9.99)(127.49,9.99)(127.49,9.99)
\psline(127.49,9.99)(134.99,10.01)
\psline(134.99,10.01)(143.14,1.87)
\psbezier(143.14,1.87)(143.14,1.87)(143.14,1.87)
\psline(143.14,1.87)(143.11,1.87)
}
\rput(131.88,5.62){$\oovee$}
\pscustom[linewidth=0.35]{\psline(195.63,1.87)(171.87,1.87)
\psline(171.87,1.87)(179.99,10.01)
\psbezier(179.99,10.01)(179.99,10.01)(179.99,10.01)
\psline(179.99,10.01)(187.49,10.01)
\psline(187.49,10.01)(195.63,1.87)
\psbezier(195.63,1.87)(195.63,1.87)(195.63,1.87)
\psline(195.63,1.87)(195.61,1.87)
}
\psline[linewidth=0.35](123.13,1.88)
(123.12,-20.62)(123.12,-26.25)
\psline[linewidth=0.35](190.62,1.88)
(190.62,-11.24)(190.62,-26.25)
\rput(183.75,5.62){$\oowedge$}
\psline[linewidth=0.35](175.62,1.88)
(175.62,-11.25)(158.12,-20.62)
\psline[linewidth=0.35](21.88,1.88)
(21.88,-11.25)(40,-21.88)
\psline[linewidth=0.35](40,48.75)
(60,39.38)(60,30)
\psline[linewidth=0.35](40,48.75)
(20.62,38.75)(20.62,30)
\psline[linewidth=0.35,border=0.3](53.12,21.25)
(53.12,8.75)(5.62,-17.5)
\rput{0}(5.62,-17.5){\psellipse[linewidth=0.35,fillstyle=solid](0,0)(1.48,-1.48)}
\end{pspicture}

\eeq
\vspace{2\baselineskip}
\end{definition}

\paragraph{Explanation.} The inputs of the morphisms on both sides of \eqref{eq:Mod} correspond to $x$ and $z$ of \eqref{eq:modular}. The equation says that the range where its left-hand side provides an output coincides with the range where its right-hand side provides an output. The right-hand morphism provides an output whenever there is $y$ such that both $x\ovee y$ and $y\owedge z$ are defined. When $\ovee$ and $\owedge$ are single-valued, then according to Lemma~\ref{lemma:sveq}, the left-hand morphism provides an output whenever $(x\ovee y)\owedge z$ and $x\ovee (y\owedge z)$ are equal. 

\begin{definition}
A convolution algebra $(A,\ovee,0,\owedge,1)$  over a self-dual object $A$ in a dagger-compact category $\CCc$ is said to satisfy the \emph{Frobenius condition} when the following equation holds
\beq\label{eq:Frob}
\newcommand{\oovee}{\ovee}
\newcommand{\oowedge}{\owedge}
\newcommand{\EQLS}{=}
\def\JPicScale{.5}
\ifx\JPicScale\undefined\def\JPicScale{1}\fi
\psset{unit=\JPicScale mm}
\psset{linewidth=0.3,dotsep=1,hatchwidth=0.3,hatchsep=1.5,shadowsize=1,dimen=middle}
\psset{dotsize=0.7 2.5,dotscale=1 1,fillcolor=black}
\psset{arrowsize=1 2,arrowlength=1,arrowinset=0.25,tbarsize=0.7 5,bracketlength=0.15,rbracketlength=0.15}
\begin{pspicture}(0,0)(87.5,20.62)
\pscustom[linewidth=0.35]{\psline(23.75,-13.13)(0,-13.12)
\psline(0,-13.12)(8.12,-5.01)
\psbezier(8.12,-5.01)(8.12,-5.01)(8.12,-5.01)
\psline(8.12,-5.01)(15.62,-4.99)
\psline(15.62,-4.99)(23.75,-13.13)
\psbezier(23.75,-13.13)(23.75,-13.13)(23.75,-13.13)
\psbezier(23.75,-13.13)(23.75,-13.13)(23.75,-13.13)
}
\rput(11.88,-9.38){$\oovee$}
\rput(36.25,0){$\EQLS$}
\psline[linewidth=0.35](18.12,20.62)(18.12,13.12)
\pscustom[linewidth=0.35]{\psline(0.01,13.14)(23.74,13.12)
\psline(23.74,13.12)(15.64,5.01)
\psbezier(15.64,5.01)(15.64,5.01)(15.64,5.01)
\psline(15.64,5.01)(8.14,4.99)
\psline(8.14,4.99)(0.01,13.14)
\psbezier(0.01,13.14)(0.01,13.14)(0.01,13.14)
\psbezier(0.01,13.14)(0.01,13.14)(0.01,13.14)
}
\psline[linewidth=0.35](11.88,5)(11.88,-5)
\psline[linewidth=0.35](19.38,-20.62)(19.38,-13.12)
\psline[linewidth=0.35](55,-20)(55,5)
\pscustom[linewidth=0.35]{\psline(73.12,4.99)(49.38,5)
\psline(49.38,5)(57.5,13.12)
\psline(57.5,13.12)(57.5,13.12)
\psline(57.5,13.12)(65,13.13)
\psline(65,13.13)(73.12,4.99)
\psbezier(73.12,4.99)(73.12,4.99)(73.12,4.99)
\psbezier(73.12,4.99)(73.12,4.99)(73.12,4.99)
}
\psline[linewidth=0.35](61.25,13.12)(61.25,20.62)
\psline[linewidth=0.35](81.88,20)(81.88,-5)
\pscustom[linewidth=0.35]{\psline(63.75,-4.99)(87.5,-5)
\psline(87.5,-5)(79.38,-13.12)
\psbezier(79.38,-13.12)(79.38,-13.12)(79.38,-13.12)
\psline(79.38,-13.12)(71.88,-13.13)
\psline(71.88,-13.13)(63.75,-4.99)
\psbezier(63.75,-4.99)(63.75,-4.99)(63.75,-4.99)
\psbezier(63.75,-4.99)(63.75,-4.99)(63.75,-4.99)
}
\psline[linewidth=0.35](75.62,-13.13)(75.62,-20)
\rput(61.25,8.75){$\oovee$}
\psline[linewidth=0.35](68.76,-5)(68.75,5)
\rput(75.62,-9.38){$\oowedge$}
\psline[linewidth=0.35](6.25,20.62)(6.25,13.12)
\psline[linewidth=0.35](6.25,-13.12)(6.25,-20.62)
\rput(11.88,9.38){$\oowedge$}
\end{pspicture}

\eeq
\vspace{1\baselineskip}
\end{definition}


The following lemma is proved by straightforward geometric transformations using the duality on $A$.

\begin{lemma}
For a convolution algebra $(A,\ovee,0,\owedge,1)$  over a self-dual object $A$ in a dagger-compact category $\CCc$, each of the following two equations is equivalent with the Frobenius condition.
\beq\label{eq:Frob-versions}
\newcommand{\oovee}{\ovee}
\newcommand{\oowedge}{\owedge}
\newcommand{\EQLS}{=}
\def\JPicScale{.47}
\ifx\JPicScale\undefined\def\JPicScale{1}\fi
\psset{unit=\JPicScale mm}
\psset{linewidth=0.3,dotsep=1,hatchwidth=0.3,hatchsep=1.5,shadowsize=1,dimen=middle}
\psset{dotsize=0.7 2.5,dotscale=1 1,fillcolor=black}
\psset{arrowsize=1 2,arrowlength=1,arrowinset=0.25,tbarsize=0.7 5,bracketlength=0.15,rbracketlength=0.15}
\begin{pspicture}(0,0)(301.26,20.62)
\rput(60,0){$\EQLS$}
\psline[linewidth=0.35](89.38,-3.75)
(89.38,-7.5)(108.75,-17.5)
\pscustom[linewidth=0.35]{\psline(93.14,-3.75)(69.37,-3.75)
\psline(69.37,-3.75)(77.49,4.37)
\psbezier(77.49,4.37)(77.49,4.37)(77.49,4.37)
\psline(77.49,4.37)(84.99,4.38)
\psline(84.99,4.38)(93.14,-3.75)
\psbezier(93.14,-3.75)(93.14,-3.75)(93.14,-3.75)
\psline(93.14,-3.75)(93.11,-3.75)
}
\rput(81.88,0){$\oovee$}
\pscustom[linewidth=0.35]{\psline(145.63,-3.75)(121.87,-3.75)
\psline(121.87,-3.75)(129.99,4.38)
\psbezier(129.99,4.38)(129.99,4.38)(129.99,4.38)
\psline(129.99,4.38)(137.49,4.38)
\psline(137.49,4.38)(145.63,-3.75)
\psbezier(145.63,-3.75)(145.63,-3.75)(145.63,-3.75)
\psline(145.63,-3.75)(145.61,-3.75)
}
\psline[linewidth=0.35](73.13,-3.75)
(73.12,-10)(73.12,-18.75)
\psline[linewidth=0.35](140.62,-3.74)
(140.62,-16.86)(140.62,-19.38)
\rput(133.75,0){$\oowedge$}
\psline[linewidth=0.35](126.25,-3.75)
(126.25,-8.12)(108.75,-17.5)
\psline[linewidth=0.35](3.12,-3.75)
(3.12,-19.38)(3.12,-13.74)
\rput(10,0){$\oovee$}
\psline[linewidth=0.35](45,-3.75)
(45,-18.75)(45,-15.62)
\rput(39.38,0){$\oowedge$}
\psline[linewidth=0.35,border=0.5](126.25,-20)
(126.25,-15.62)(88.75,13.12)
\pscustom[linewidth=0.35]{\psline(21.88,-3.77)(-1.88,-3.75)
\psline(-1.88,-3.75)(6.25,4.38)
\psbezier(6.25,4.38)(6.25,4.38)(6.25,4.38)
\psline(6.25,4.38)(13.75,4.38)
\psline(13.75,4.38)(21.88,-3.77)
\psbezier(21.88,-3.77)(21.88,-3.77)(21.88,-3.77)
\psbezier(21.88,-3.77)(21.88,-3.77)(21.88,-3.77)
}
\pscustom[linewidth=0.35]{\psline(50.64,-3.77)(26.88,-3.75)
\psline(26.88,-3.75)(35,4.38)
\psbezier(35,4.38)(35,4.38)(35,4.38)
\psline(35,4.38)(42.5,4.38)
\psline(42.5,4.38)(50.64,-3.77)
\psbezier(50.64,-3.77)(50.64,-3.77)(50.64,-3.77)
\psline(50.64,-3.77)(50.62,-3.77)
}
\psline[linewidth=0.35](25,15)
(39.38,8.14)(39.38,4.38)
\psline[linewidth=0.35](25,15)
(10,7.49)(10,4.38)
\psline[linewidth=0.35](16.25,-3.75)
(16.25,-19.38)(16.25,-13.74)
\psline[linewidth=0.35](32.5,-3.75)
(32.5,-19.38)(32.5,-15.63)
\psline[linewidth=0.35](88.75,13.12)
(81.25,10)(81.25,4.38)
\psline[linewidth=0.35,border=0.5](88.75,-19.38)
(88.75,-15.62)(126.25,13.12)
\psline[linewidth=0.35](126.25,13.12)
(133.75,10)(133.75,4.38)
\rput(235.63,0){$\EQLS$}
\psline[linewidth=0.35](265,-3.75)
(265,-7.5)(273.75,-11.88)
\pscustom[linewidth=0.35]{\psline(268.76,-3.75)(245,-3.75)
\psline(245,-3.75)(253.12,4.37)
\psbezier(253.12,4.37)(253.12,4.37)(253.12,4.37)
\psline(253.12,4.37)(260.62,4.38)
\psline(260.62,4.38)(268.76,-3.75)
\psbezier(268.76,-3.75)(268.76,-3.75)(268.76,-3.75)
\psline(268.76,-3.75)(268.74,-3.75)
}
\rput(257.5,0){$\oovee$}
\pscustom[linewidth=0.35]{\psline(301.26,-3.75)(277.5,-3.75)
\psline(277.5,-3.75)(285.62,4.38)
\psbezier(285.62,4.38)(285.62,4.38)(285.62,4.38)
\psline(285.62,4.38)(293.12,4.38)
\psline(293.12,4.38)(301.26,-3.75)
\psbezier(301.26,-3.75)(301.26,-3.75)(301.26,-3.75)
\psline(301.26,-3.75)(301.24,-3.75)
}
\psline[linewidth=0.35](248.76,-3.75)
(248.74,-10)(248.74,-18.75)
\psline[linewidth=0.35](296.24,-3.74)
(296.24,-16.86)(296.24,-19.38)
\rput(289.38,0){$\oowedge$}
\psline[linewidth=0.35](281.88,-3.75)
(281.88,-8.12)(273.75,-11.88)
\psline[linewidth=0.35](257.5,20)
(257.5,4.37)(257.5,10)
\rput(176.25,0){$\oovee$}
\psline[linewidth=0.35](290,20)
(290,5)(290,8.12)
\rput(215.63,0){$\oowedge$}
\pscustom[linewidth=0.35]{\psline(188.13,-3.77)(164.37,-3.75)
\psline(164.37,-3.75)(172.51,4.38)
\psbezier(172.51,4.38)(172.51,4.38)(172.51,4.38)
\psline(172.51,4.38)(180.01,4.38)
\psline(180.01,4.38)(188.13,-3.77)
\psbezier(188.13,-3.77)(188.13,-3.77)(188.13,-3.77)
\psbezier(188.13,-3.77)(188.13,-3.77)(188.13,-3.77)
}
\pscustom[linewidth=0.35]{\psline(226.89,-3.77)(203.13,-3.75)
\psline(203.13,-3.75)(211.25,4.38)
\psbezier(211.25,4.38)(211.25,4.38)(211.25,4.38)
\psline(211.25,4.38)(218.75,4.38)
\psline(218.75,4.38)(226.89,-3.77)
\psbezier(226.89,-3.77)(226.89,-3.77)(226.89,-3.77)
\psline(226.89,-3.77)(226.87,-3.77)
}
\psline[linewidth=0.35](195.63,18.75)
(215.63,9.38)(215.63,4.38)
\psline[linewidth=0.35](195.63,18.75)
(176.25,8.75)(176.25,4.38)
\psline[linewidth=0.35](168.75,-3.75)
(168.75,-19.38)(168.75,-13.74)
\psline[linewidth=0.35](221.87,-3.75)
(221.87,-19.38)(221.87,-15.63)
\psline[linewidth=0.35,border=0.5](215.63,20.62)
(215.63,15.62)(187.5,-12.5)
\psline[linewidth=0.35](187.5,-12.5)
(183.75,-8.75)(183.75,-3.75)
\psline[linewidth=0.35,border=0.5](176.25,20)
(176.25,15)(204.38,-13.12)
\psline[linewidth=0.35](204.38,-13.75)
(208.75,-9.38)(208.75,-3.75)
\end{pspicture}

\eeq
\vspace{1\baselineskip}
\end{lemma}

\begin{lemma}
If the convolution algebra $(A,\ovee,0,\owedge,1)$  over a self-dual object $A$ in a dagger-compact category $\CCc$ consists of single-valued operations, then the Frobenius condition is also equivalent with equation \eqref{eq:Mod}.
\end{lemma}

\begin{proof}
We use Lemma~\ref{lemma:sveq}. Let $f$ be the right-hand side of the second equation of \eqref{eq:Frob-versions}; let $g$ be the left-hand side of \eqref{eq:Frob-versions}. Lemma~\ref{lemma:sveq} says that $f=g$ if and only if $\cun\circ\left((g^\ddag \circ f) \convv \id\right) = f\convv g = \cun \circ f$. But it is easy to see that $\cun\circ\left((g^\ddag \circ f) \convv \id\right)$ reduces to the  left-hand side of \eqref{eq:Mod}, whereas $\cun \circ f$ is the right hand side of \eqref{eq:Mod}. Equation \eqref{eq:Mod} thus holds if and only if the second equation of \eqref{eq:Frob-versions} holds.
\end{proof}

\paragraph{Remark.} The correspondence between the modularity and the Frobenius condition is reflected in the geometry of the left-hand diagram of \eqref{eq:Mod}, as displayed on the next figure.

\vspace{1\baselineskip}
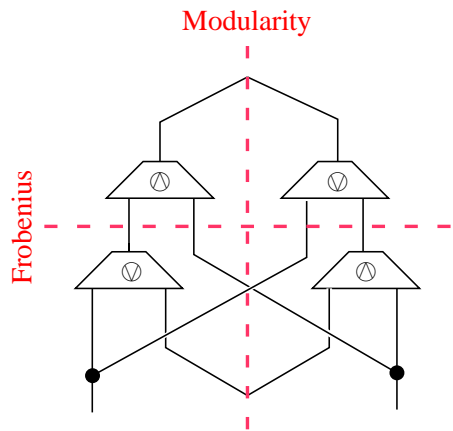
\[
\newcommand{\oovee}{\ovee}
\newcommand{\oowedge}{\owedge}
\newcommand{\Modularity}{\color{red} Modularity}
\newcommand{\Frobenius}{\color{red} Frobenius}
\newcommand{\EQLS}{=}
\def\JPicScale{.6}
\ifx\JPicScale\undefined\def\JPicScale{1}\fi
\psset{unit=\JPicScale mm}
\psset{linewidth=0.3,dotsep=1,hatchwidth=0.3,hatchsep=1.5,shadowsize=1,dimen=middle}
\psset{dotsize=0.7 2.5,dotscale=1 1,fillcolor=black}
\psset{arrowsize=1 2,arrowlength=1,arrowinset=0.25,tbarsize=0.7 5,bracketlength=0.15,rbracketlength=0.15}
\begin{pspicture}(0,0)(100,75)
\psline[linewidth=0.35](28.75,36.25)
(28.75,24.38)(28.75,26.26)
\rput(75,40){$\oovee$}
\pscustom[linewidth=0.35]{\psline(80.62,36.25)(80.62,24.38)
\psbezier(80.62,24.38)(80.62,24.38)(80.62,24.38)
}
\pscustom[linewidth=0.35]{\psline(40.64,16.24)(16.87,16.24)
\psline(16.87,16.24)(24.99,24.37)
\psbezier(24.99,24.37)(24.99,24.37)(24.99,24.37)
\psline(24.99,24.37)(32.49,24.38)
\psline(32.49,24.38)(40.64,16.24)
\psbezier(40.64,16.24)(40.64,16.24)(40.64,16.24)
\psline(40.64,16.24)(40.61,16.24)
}
\rput(29.38,20){$\oovee$}
\pscustom[linewidth=0.35]{\psline(93.13,16.24)(69.37,16.24)
\psline(69.37,16.24)(77.49,24.38)
\psbezier(77.49,24.38)(77.49,24.38)(77.49,24.38)
\psline(77.49,24.38)(84.99,24.38)
\psline(84.99,24.38)(93.13,16.24)
\psbezier(93.13,16.24)(93.13,16.24)(93.13,16.24)
\psline(93.13,16.24)(93.11,16.24)
}
\rput(35.62,40){$\oowedge$}
\psline[linewidth=0.35](20.63,16.26)
(20.62,-6.25)(20.62,-11.25)
\psline[linewidth=0.35](88.12,16.26)
(88.12,3.13)(88.12,-10.62)
\rput(81.25,20){$\oowedge$}
\psline[linewidth=0.35,border=0.45](73.12,16.26)
(73.12,3.12)(55,-7.5)
\psline[linewidth=0.35,border=0.5](43.12,36.25)
(43.12,23.76)(88.12,-2.5)
\rput{0}(88.12,-2.5){\psellipse[linewidth=0.35,fillstyle=solid](0,0)(1.48,-1.48)}
\pscustom[linewidth=0.35]{\psline(47.51,36.24)(23.75,36.25)
\psline(23.75,36.25)(31.88,44.38)
\psbezier(31.88,44.38)(31.88,44.38)(31.88,44.38)
\psline(31.88,44.38)(39.38,44.38)
\psline(39.38,44.38)(47.51,36.24)
\psbezier(47.51,36.24)(47.51,36.24)(47.51,36.24)
\psline(47.51,36.24)(47.5,36.24)
}
\pscustom[linewidth=0.35]{\psline(86.26,36.24)(62.5,36.25)
\psline(62.5,36.25)(70.63,44.38)
\psbezier(70.63,44.38)(70.63,44.38)(70.63,44.38)
\psline(70.63,44.38)(78.13,44.38)
\psline(78.13,44.38)(86.26,36.24)
\psbezier(86.26,36.24)(86.26,36.24)(86.26,36.24)
\psline(86.26,36.24)(86.25,36.24)
}
\psline[linewidth=0.35](36.88,16.26)
(36.88,3.12)(55,-7.5)
\psline[linewidth=0.35](55,63.12)
(75,53.76)(75,44.38)
\psline[linewidth=0.35](55,63.12)
(35.62,53.12)(35.62,44.38)
\psline[linewidth=0.35,border=0.3](68.12,35.62)
(68.12,23.12)(20.62,-3.12)
\rput{0}(20.62,-3.12){\psellipse[linewidth=0.35,fillstyle=solid](0,0)(1.48,-1.48)}
\newrgbcolor{userLineColour}{1 0.2 0.4}
\psline[linewidth=0.75,linecolor=userLineColour,linestyle=dashed,dash=2.8 4.2](10,30)(100,30)
\newrgbcolor{userLineColour}{1 0.2 0.4}
\psline[linewidth=0.75,linecolor=userLineColour,linestyle=dashed,dash=2.8 4.2](55,-15)(55,70)
\rput(55,75){\Modularity}
\rput{90}(5,30){\Frobenius}
\end{pspicture}

\]

\vspace{2\baselineskip}
\noindent The vertical line splits the diagram into two sides of the modularity condition. Mutatis mutandis, the horizontal line through the middle of this diagram splits it into two sides of the Frobenius condition.

\begin{corollary}\label{prop:modular}
A superspecial algebra $(A,\ovee,\owedge,0,1,\ortt)$ over a self-dual object $A$ in a dagger-compact category $\CCc$ satisfies the Frobenius condition if and only if it is modular.
\end{corollary}
  

\section{Further work}\label{Sec:Conclusion}
The first task is to extend the correspondences between (modular) effect algebras and (Frobenius) superspecial algebras, spelled out in Propositions~\ref{prop:eff} and \ref{prop:modular} into functors between the corresponding categories. The different components were built into these different structures to capture different concepts. The fact that these different conceptual components, when combined, lead to equivalent categories suggests that there are underlying conceptual connections that may be of interest. What is the connection between the entanglement type of the $W$-state, realized by the antispecial law on one side, and the sharpness of the units of the effect algebra operations on the other side? 

Another immediate task is to lift the characterization of (modular) effect algebras as (Frobenius) superspecial algebras from the concrete category $\Rel$ of sets and relations, where effect algebras seem to normally live, to the abstract framework of dagger-compact categories, where the usual pointwise definition of effect algebras cannot be stated. If we \emph{define}\/ an effect algebra in a dagger-compact category to be a superspecial algebra, then the convenient and intuitive language of effect algebras (suitably extended by the scalar factors, which are trivial in $\Rel$) becomes available not only in the richer nonstandard models of quantum mechanics \cite{PavlovicD:QPL09}, but ironically even in the standard Hilbert space model, whose relevant features were originally intended to be separated from the irrelevant ones by the language of effect algebras.  

Last but not least, since every superspecial Frobenius algebra implements a GHZ/W-pair of \cite{Coecke-Kissinger:anti}, and every GHZ/W-pair implements a Z/X-pair of complementary observables, modular effect algebras in any of these frameworks may provide a useful new mathematical interface to complementary observables.


\bibliographystyle{eptcs}
\bibliography{PavlovicD,ref-antieffect,CT}
\end{document}